\pdfoutput=1
\documentclass{lmcs} 


\usepackage{hyperref}
\usepackage{mathtools}

\usepackage{amssymb}
\usepackage{stmaryrd}
\usepackage[ruled,vlined,linesnumbered]{algorithm2e}
\usepackage{cleveref}
\usepackage{csquotes}
\usepackage{graphicx}
\usepackage{relsize}
\usepackage{xspace}
\usepackage{thm-restate}
\usepackage{nicefrac}
\usepackage{tabularx}
\usepackage[textwidth=4cm]{todonotes} 
\usepackage{xpatch}
\usepackage{booktabs}
\usepackage{flafter}

\theoremstyle{plain}\newtheorem{theorem}[thm]{Theorem}
\theoremstyle{plain}\newtheorem{example}[thm]{Example}
\theoremstyle{plain}\newtheorem{definition}[thm]{Definition}
\theoremstyle{plain}\newtheorem{lemma}[thm]{Lemma}
\theoremstyle{plain}\newtheorem{remark}[thm]{Remark}
\crefname{prop}{proposition}{propositions}
\Crefname{prop}{Proposition}{Propositions}
\crefname{cor}{corollary}{corollaries}
\Crefname{cor}{Corollary}{Corollaries}

\newcommand{\eeq}{~{}={}~}
\newcommand{\mylambda}[1]{\ensuremath{\lambda #1.\,}}

\newcommand{\finishexample}{\ensuremath{\triangle}}
\newcommand{\finishdefinition}{\ensuremath{\triangle}}

\newcommand{\toolfont}[1]{\textnormal{\textsc{#1}}}

\newcommand{\theory}{\ensuremath{\mathcal{T}}}

\newcommand{\inlightgray}{\color{gray}}

\newcommand{\AssignSymbol}{\mathrel{\textnormal{\texttt{:=}}}}
\newcommand{\ASSIGN}[2]{\ensuremath{#1 \AssignSymbol #2}}

\newcommand{\NN}{\ensuremath{\mathbb{N}}}
\newcommand{\NNO}{\ensuremath{\mathbb{N}_0}}
\newcommand{\nna}{\ensuremath{n}}
\newcommand{\nnb}{\ensuremath{m}}
\newcommand{\nnc}{\ensuremath{k}}

\newcommand{\QQ}{\ensuremath{\mathbb{Q}}}
\newcommand{\ExtQQ}{\ensuremath{\mathbb{Q}^{\pm\infty}}}
\newcommand{\rata}{\ensuremath{q}}
\newcommand{\ratb}{\ensuremath{p}}

\newcommand{\RR}{\ensuremath{\mathbb{R}}}
\newcommand{\ExtRR}{\ensuremath{\mathbb{R}^{\pm\infty}}}
\newcommand{\reala}{\ensuremath{r}}

\newcommand{\PosReals}{\RR_{\geq 0}}
\newcommand{\PosRealsInf}{\PosReals^{\infty}}

\newcommand{\proba}{\ensuremath{p}}

\newcommand{\idxa}{\ensuremath{i}}

\newcommand{\idxc}{\ensuremath{j}}
\newcommand{\idxd}{\ensuremath{j'}}

\newcommand{\Vars}{\ensuremath{\mathsf{Vars}}}
\newcommand{\FVars}[1]{\ensuremath{\mathsf{FV} \left( #1 \right)}}
\newcommand{\vara}{\ensuremath{x}}
\newcommand{\varb}{\ensuremath{y}}
\newcommand{\varc}{\ensuremath{z}}
\newcommand{\cardnum}{\ensuremath{\vert\Vars\vert}}

\newcommand{\LinAX}{\ensuremath{\mathsf{AffExp}}}
\newcommand{\axa}{\ensuremath{a}}

\newcommand{\axc}{\ensuremath{e}}
\newcommand{\ExtLinAX}{\ensuremath{\LinAX^{\pm\infty}}}
\newcommand{\eaxa}{\ensuremath{\tilde{a}}}
\newcommand{\eaxb}{\ensuremath{\tilde{b}}}
\newcommand{\eaxc}{\ensuremath{\tilde{e}}}

\newcommand{\eaxl}{\ensuremath{\tilde{\ell}}}
\newcommand{\eaxu}{\ensuremath{\tilde{u}}}

\newcommand{\Bool}{\ensuremath{\mathsf{Bool}}}
\newcommand{\true}{\ensuremath{\mathsf{true}}}
\newcommand{\false}{\ensuremath{\mathsf{false}}}
\newcommand{\bxa}{\ensuremath{\varphi}}
\newcommand{\bxb}{\ensuremath{\psi}}
\newcommand{\bxc}{\ensuremath{\xi}}

\newcommand{\DNF}{\text{{DNF}}}
\newcommand{\bimplies}{\ensuremath{\longrightarrow}}

\newcommand{\LQ}{\ensuremath{\mathsf{LinQuant}}}
\newcommand{\LE}{\ensuremath{\mathsf{LinExp}}}
\newcommand{\lqa}{\ensuremath{f}}
\newcommand{\lqb}{\ensuremath{f'}}
\newcommand{\lqc}{\ensuremath{g}}
\newcommand{\lqd}{\ensuremath{g'}}
\newcommand{\qflqa}{\ensuremath{h}}
\newcommand{\qflqb}{\ensuremath{h'}}
\newcommand{\quanta}{\ensuremath{Q}}

\newcommand{\XGNF}[1]{\ensuremath{\mathsf{GNF}_#1}}

\newcommand{\iverson}[1]{\ensuremath{\left[ #1 \right]}}
\newcommand{\SNF}{\ensuremath{\sum\limits_{\idxa=1}^{\nna}\iverson{\bxa_\idxa}\cdot\eaxa_\idxa}}

\newcommand{\PS}{\ensuremath{\Sigma}}
\newcommand{\psa}{\ensuremath{\sigma}}

\newcommand{\statesubst}[2]{[#1\mapsto#2]}

\newcommand{\Sup}{\reflectbox{\textnormal{\textsf{\fontfamily{phv}\selectfont S}}}\hspace{.2ex}}
\newcommand{\SupOO}[2]{\Sup#1\colon#2}
\newcommand{\Inf}{\raisebox{.6\depth}{\rotatebox{-30}{\textnormal{\textsf{\fontfamily{phv}\selectfont \reflectbox{J}}}}\hspace{-.1ex}}}
\newcommand{\InfOO}[2]{\Inf#1\colon#2}

\newcommand{\syntRepl}[3]{\ensuremath{#2[#1/#3]}}

\newcommand{\sem}[2]{{}^{#1}\llbracket#2\rrbracket}

\newcommand{\qentails}{\ensuremath{\models}}
\newcommand{\qqentails}{\ensuremath{~{}\models{}~}}
\newcommand{\lleq}{\ensuremath{~{}\leq{}}}

\newcommand{\disj}{\ensuremath{D}}

\newcommand{\lit}{\ensuremath{L}}

\newcommand{\Bounds}[3]{\ensuremath{\mathsf{Bnd}_{#2 {}#1{}}(#3)}}
\newcommand{\LBounds}[1]{\ensuremath{\mathsf{LBnd}_{#1}}}
\newcommand{\UBounds}[1]{\ensuremath{\mathsf{UBnd}_{#1}}}
\newcommand{\bounds}{\ensuremath{\{>,\geq,<,\leq\}}}

\newcommand{\phiSAT}{\ensuremath{\bxa_{\exists}}}
\newcommand{\phiInf}{\ensuremath{\bxa_{\inf}}}
\newcommand{\phiSup}{\ensuremath{\bxa_{\sup}}}

\newcommand{\mina}{\ensuremath{m_1}}
\newcommand{\minb}{\ensuremath{m_2}}
\newcommand{\maxa}{\ensuremath{M_1}}
\newcommand{\maxb}{\ensuremath{M_2}}

\newcommand{\down}[2]{[#1]\cdot#2+[\lnot#1]\cdot(-\infty)}
\newcommand{\up}[2]{[#1]\cdot#2+[\lnot#1]\cdot\infty}

\newcommand{\MAX}{\ensuremath{\mathsf{MAX}}}
\newcommand{\MIN}{\ensuremath{\mathsf{MIN}}}

\newcommand{\seta}{\ensuremath{M}}

\newcommand{\QE}[1]{\ensuremath{\mathsf{QE}(#1)}}

\newcommand{\ELIMnobrace}{\ensuremath{\mathsf{Elim}}}
\newcommand{\ELIM}[1]{\ensuremath{\ELIMnobrace(#1)}}

\newcommand{\lqwidth}[1]{\ensuremath{| #1 |_{\rightarrow}}}
\newcommand{\lqdepth}[1]{\ensuremath{| #1 |}_{\downarrow}}
\newcommand{\bxsize}[1]{\ensuremath{|#1|}}

\newcommand{\pgcl}{\ensuremath{\mathsf{pGCL}}}

\newcommand{\pp}{\ensuremath{C}}

\newcommand{\SKIP}{\ensuremath{\textnormal{\texttt{skip}}}}

\newcommand{\COMPOSE}[2]{\ensuremath{{#1}{\,;}~ {#2}}}

\newcommand{\PCHOICE}[3]{\ensuremath{\left\{\, {#1} \,\right\}\mathrel{\left[\,#2\,\right]}\left\{\, {#3} \,\right\}}}

\newcommand{\NDCHOICE}[2]{\ensuremath{\left\{\, {#1} \,\right\}\mathrel{\Box}\left\{\, {#2} \,\right\}}}
\newcommand{\UNDCHOICE}[1]{\ASSIGN{#1}{\QQ}}

\newcommand{\IFSYMBOL}{\ensuremath{\textnormal{\texttt{if}}}}

\newcommand{\ELSESYMBOL}{\ensuremath{\textnormal{\texttt{else}}}}

\newcommand{\ITE}[3]{\ensuremath{\IFSYMBOL\,\left(\, {#1} \,\right)\,\left\{\, {#2} \,\right\}\,\ELSESYMBOL\,\left\{\, {#3} \,\right\}}}

\newcommand{\WHILESYMBOL}{\ensuremath{\textnormal{\texttt{while}}}}
\newcommand{\WHILE}[1]{\ensuremath{\WHILESYMBOL \left(\, {#1} \,\right)\left\{\right.}}
\newcommand{\WHILEDO}[2]{\ensuremath{\WHILESYMBOL \left(\, {#1} \,\right)\left\{\, {#2} \,\right\}}}

\newcommand{\E}{\mathbb{E}} 

\newcommand{\FF}{\ensuremath{X}}
\newcommand{\FG}{\ensuremath{Y}}

\newcommand{\xawpsymbol}{\overline{\mathsf{awp}}}
\newcommand{\xawptrans}[1]{\xawpsymbol\llbracket#1\rrbracket}
\newcommand{\xawp}[2]{\xawptrans{#1}\left(#2\right)}

\newcommand{\xdwpsymbol}{\mathsf{d\overline{wp}}}
\newcommand{\xdwptrans}[1]{\xdwpsymbol\llbracket#1\rrbracket}
\newcommand{\xdwp}[2]{\xdwptrans{#1}\left(#2\right)}

\newcommand{\dwpsymbol}{\mathsf{dwp}}
\newcommand{\dwptrans}[1]{\dwpsymbol\llbracket#1\rrbracket}
\newcommand{\dwp}[2]{\dwptrans{#1}\left(#2\right)}

\newcommand{\awpsymbol}{\mathsf{awp}}
\newcommand{\awptrans}[1]{\awpsymbol\llbracket#1\rrbracket}
\newcommand{\awp}[2]{\awptrans{#1}\left(#2\right)}

\newcommand{\somewpsymbol}{\ensuremath{\mathcal{T}}}
\newcommand{\somewptrans}[1]{\somewpsymbol\llbracket#1\rrbracket}
\newcommand{\somewp}[2]{\somewptrans{#1}\left(#2\right)}

\newcommand{\xsomewpsymbol}{\ensuremath{\overline{\mathcal{T}}}}
\newcommand{\xsomewptrans}[1]{\xsomewpsymbol\llbracket#1\rrbracket}
\newcommand{\xsomewp}[2]{\xsomewptrans{#1}\left(#2\right)}
\newcommand{\Norm}[1]{\ensuremath{\mathsf{Norm}(#1)}}
\newcommand{\PART}[1]{\ensuremath{\mathsf{Part}(#1)}}

\newcommand{\EAll}[1]{\ensuremath{\ELIM{\PART{#1}}}}
\newcommand{\PosRatsInf}{\ensuremath{\QQ_{\geq 0}\cup\{\infty\}}}
\newcommand{\Feas}[2]{\ensuremath{S_{#1}(#2)}}

\begin{document}
		
		\title[Quantifier Elimination and Craig Interpolation, Quantitatively]{Quantifier Elimination and Craig Interpolation, \\ Quantitatively
			}
		\titlecomment{{\lsuper*}This article extends our conference paper \cite{DBLP:conf/fossacs/BatzKO25}.}
		
		\author[K.~Batz]{Kevin Batz \lmcsorcid{0000−0001−8705−2564}}[a]
				\author[J.-P.~Katoen]{Joost-Pieter Katoen\lmcsorcid{0000-0002-6143-1926}}[b]
		\author[N.~Orhan]{Nora Orhan}[b]
		
		\address{University of Münster, Germany}	
		\email{kevin.batz@uni-muenster.de}  
		
		\address{RWTH Aachen University, Germany}	
		\email{nora.hiseni@rwth-aachen.de, katoen@cs.rwth-aachen.de}  
		
		
		
		
		
		\begin{abstract}
			\noindent
			Quantifier elimination (QE) and Craig interpolation (CI) are central to various state-of-the-art automated approaches to hardware- and software verification. They are rooted in the Boolean setting and are successful for, e.g., first-order theories such as linear rational arithmetic. What about their applicability in the quantitative setting where formulae evaluate to numbers and quantitative supremum/infimum quantifiers are the natural counterpart to Boolean quantifiers? Applications include establishing quantitative properties of programs such as bounds on expected outcomes of probabilistic programs featuring non-determinism and analyzing the flow of information through programs.
			
			In this paper, we present the --- to the best of our knowledge --- first QE algorithm for possibly unbounded, $\infty$- or $(-\infty)$-valued, or discontinuous \emph{piecewise linear quantities}. They are the quantitative counterpart to linear rational arithmetic, and they are a popular quantitative assertion language for probabilistic program verification. We provide rigorous soundness proofs as well as upper bounds on the size of the output. Moreover, we present two applications of our QE algorithm. First, our algorithm yields a quantitative CI theorem: Given arbitrary piecewise linear quantities $\lqa,\lqc$ with $\lqa\qentails\lqc$ (i.e., $\lqa$ is at most $\lqc$ at every variable valuation), both the strongest and the weakest Craig interpolants are quantifier-free and effectively constructible. Such interpolants lie between the two quantities and use only their common free variables. Second, we apply our QE algorithm to computing minimal and maximal expected outcomes of loop-free probabilistic programs featuring unbounded non-determinism, and to checking inductivity of candidate upper bounds on minimal and maximal expected outcomes of probabilistic loops featuring unbounded non-determinism.
		\end{abstract}
		
		\maketitle
		
		\section{Introduction}
		\label{sec:intro}
		\emph{Quantifier elimination} algorithms transform a first-order formula $\varphi$ over some background theory $\theory$ into a quantifier-free formula $\QE{\varphi}$ equivalent to $\varphi$ modulo $\theory$. Prime examples include Fourier-Motzkin variable elimination \cite{motzkin1936beitraege,fourier1825analyse} and virtual substitution \cite{DBLP:journals/cj/LoosW93} for linear rational arithmetic (LRA), Cooper's method \cite{cooper} for linear integer arithmetic, and Cylindrical Algebraic Decomposition \cite{DBLP:conf/automata/Collins75} for non-linear real arithmetic. Quantifier elimination is extensively leveraged by automatic hard- and software verification techniques for, e.g., computing images of state sets \cite{DBLP:conf/cav/KomuravelliGC14,DBLP:journals/fmsd/KomuravelliGC16}, or for synthesizing loop invariants either from templates \cite{DBLP:conf/qest/GretzKM13,DBLP:conf/cav/ColonSS03,Kapur2005AutomaticallyGL} or by solving abduction problems \cite{DBLP:conf/oopsla/DilligDLM13}.

\emph{Craig interpolation \cite{Craig_1957}} 
is vital to various automatic hard- and software verification techniques. A first-order theory $\theory$ is called \emph{(quantifier-free) interpolating} \cite{10.1145/1181775.1181789}, if for all formulae $\varphi,\psi$ with $\varphi \models_\theory \psi$, there is an effectively constructible (quantifier-free) formula $\vartheta$ --- called \emph{Craig interpolant} of $(\varphi,\psi)$ --- with $\varphi \models_\theory \vartheta \models_\theory \psi$ and such that all variables occurring free in $\vartheta$ also occur free in \emph{both} $\varphi$ and $\psi$. Intuitively, $\varphi \models_\theory \psi$ encodes some (desirable or undesirable) reachability information and $\vartheta$ is a concise explanation of this fact, abstracting away irrelevant details. The computation of (quantifier-free) Craig interpolants is a vivid area of research \cite{DBLP:conf/cav/AlbarghouthiM13,DBLP:conf/cade/GanDXZKC16,DBLP:conf/cade/ChenWAZKZ19,DBLP:conf/fm/WuWXLZG24,DBLP:conf/fossacs/CateC24} with applications ranging from symbolic finite-state model checking \cite{DBLP:conf/cav/McMillan03,DBLP:conf/cav/VizelG14,DBLP:conf/cav/KrishnanVGG19} over computing transition power abstractions \cite{DBLP:conf/fm/BritikovBSF24} to automatic infinite-state software verification \cite{DBLP:conf/popl/HenzingerJMM04,craig_to_software,DBLP:conf/tacas/JhalaM06,DBLP:conf/hvc/SeryFS11,DBLP:conf/fmcad/SeryFS12,DBLP:conf/cav/McMillan06,DBLP:conf/lpar/AlbertiBGRS12}.

\emph{Quantitative program verification} includes reasoning about expected outcomes of probabilistic programs via \emph{weakest pre-expectations}, which describe expected final outcomes as functions of the initial state \cite{kozen83,kozen85,mciver_morgan,kaminski_diss}, reasoning about the quantitative flow of information via \emph{quantitative strongest post} \cite{strongest_post}, and reasoning about competitive ratios of online algorithms via \emph{weighted programming} \cite{weighted}. Quantitative reasoning requires a shift:
 Predicates, i.e., maps from program states to truth values, are replaced by \emph{quantities}%
\footnote{We adopt this term from Zhang and Kaminski \cite{strongest_post}. In the realm of weakest pre-expectations, quantitative assertions are usually referred to as \emph{expectations} \cite{mciver_morgan}. In weighted programming, they are called \emph{weightings} \cite{weighted}.}, 
which map program states to \emph{extended reals} in $\ExtRR\eeq\RR\cup\{-\infty,\infty\}$. Here a program state is a \emph{variable valuation}, assigning a rational number in $\QQ$ to each variable. 

The classical quantifiers \enquote{there exists} $\exists$ and \enquote{for all} $\forall$ from predicate logic are replaced by \emph{quantitative} supremum $\Sup$ and infimum $\Inf$ quantifiers \cite{DBLP:journals/pacmpl/BatzKKM21}. They take, respectively, the least upper bound and greatest lower bound of the values obtained by varying the quantified variable. These quantifiers occur naturally when reasoning with quantitative program logics: Very much like \emph{classical} strongest post-\emph{conditions} introduce an $\exists$-quantifier for assignments \cite{DBLP:books/daglib/0067387}, \emph{quantitative} strongest post introduces a $\Sup$-quantifier (cf.\ \cite[Table 2]{strongest_post}). Similarly, whereas \emph{classical} weakest pre-\emph{conditions} introduce a $\forall$-quantifier for demonically resolving unbounded non-determinism of the form $\ASSIGN{\vara}{\QQ}$ (read: assign to $\vara$ an \emph{arbitrary} rational number) \cite{dijkstra_discipline,mullerBuildingDeductiveProgram2019}, \emph{quantitative} weakest pre-\emph{expectations} introduce a $\Inf$-quantifier \cite{qsl_popl,caesar}.
%
\begin{example}
	\label{ex:intro}
	In this paper, we focus on \emph{piecewise linear quantities}, whose quantifier-free expressions are sums of affine expressions or constant infinities, selected by Boolean conditions called \emph{guards}. The \emph{Iverson bracket} $\iverson{\bxa}$ is $1$ when the guard $\bxa$ holds and $0$ otherwise. For example, 
	\begin{align*}
		 \lqc \eeq 
		\inlightgray \underbrace{
		 	\color{black}[\inlightgray \underbrace{ \color{black}\varb_1 \geq \varc \bimplies (\vara - 2 < \varb_1
		 	\wedge -\vara \geq \varb_3 \wedge \vara \geq \varb_2)}_{{} =\, \bxa}
		 	]
		 	\cdot {\color{black}(} \underbrace{ \vphantom{(}
		 		\color{black} 2\cdot \vara + \varc}_{{} =\, \eaxa} {\color{black})} 
		 }_{\text{evaluate to $\eaxa$ on variable valuation $\psa$ if $\psa \models \bxa$, and to $0$ otherwise}}~{\color{black} ,}
	\end{align*}
	where $\vara,\varb_1,\ldots$ are $\QQ$-valued variables. We can think of $\lqc$ as a formula that evaluates to extended rationals from $\QQ \cup \{-\infty,\infty\}$ instead of truth values. Write $\sem{\psa}{\lqc}$ for the value of $\lqc$ at $\psa$, and $\psa\statesubst{\vara}{\rata}$ for the valuation that changes only $\vara$ to $\rata$. By prefixing $\lqc$ with, e.g., a supremum quantifier, we obtain a new piecewise linear quantity $\SupOO{\vara}{\lqc}$, which, on variable valuation $\psa$, evaluates to the supremum of $\lqc$ under all possible values for $\vara$, i.e.,
	\begin{align*}
		\sem{\psa}{\SupOO{\vara}{\lqc}} \eeq  \sup\Bigl\{\sem{\psa\statesubst{\vara}{\rata}}{\lqc}\mid\rata\in\QQ\Bigr\} ~.
		%
\tag*{\finishexample}
	\end{align*}
\end{example}

\paragraph{Our Contribution: Quantitative Quantifier Elimination and Craig Interpolation.}
Piecewise linear quantities over $\QQ$-valued variables are to quantitative probabilistic program verification what first-order linear rational arithmetic is to classical program verification: Their entailment problem, i.e., 
\begin{center}
	\emph{Given piecewise linear quantities $\lqa$ and $\lqb$}, \\ \emph{does $\inlightgray\underbrace{\color{black}\lqa\qentails\lqb}_{\text{for all variable valuations $\psa$}\colon \sem{\psa}{\lqa} \leq \sem{\psa}{\lqb}}$ hold?}
\end{center}
 is decidable \cite{DBLP:conf/sas/KatoenMMM10,DBLP:conf/cav/BatzCKKMS20}, they are effectively closed under weakest pre-expectations of loop-free linear probabilistic programs \cite{DBLP:conf/cav/BatzCKKMS20}, and they have been shown to be sufficiently expressive for the verification of various probabilistic programs \cite{DBLP:conf/tacas/BatzCJKKM23,DBLP:conf/cav/BatzCKKMS20,DBLP:journals/pacmpl/YangFFZKW26,DBLP:conf/pldi/NgoC018,DBLP:conf/cav/ChakarovS13}. This renders piecewise linear quantities one of the most prevalent quantitative assertion languages {in automatic reasoning over probabilistic programs.}

Reasoning with piecewise linear quantities containing the quantitative $\Sup$ or $\Inf$ quantifiers has, however, received scarce attention, let alone algorithmically. The field of \emph{quantitative Craig interpolation} is rather unexplored, as well. The goal of this paper is to lay the foundations (i) for developing quantitative quantifier elimination- and Craig interpolation-based approaches to automatic quantitative program verification and (ii) for simplifying the reasoning with quantitative assertions involving quantitative quantifiers. Towards this end:
\begin{enumerate}
	\item We contribute the --- to the best of our knowledge --- first quantitative quantifier elimination algorithm for \emph{possibly unbounded, $\infty$ or $-(\infty)$-valued, or discontinuous} piecewise linear quantities. More formally, given an \emph{arbitrary} piecewise linear quantity $\lqa$ possibly containing quantitative quantifiers, our algorithm computes a \emph{quantifier-free equivalent of $\lqa$.} For $\SupOO{\vara}{\lqc}$ from \Cref{ex:intro}, i.e.,
	\[
		\SupOO{\vara}{\lqc} 
		\quad{}={}\quad
		 \SupOO{\vara}{\iverson{\varb_1 \geq \varc \bimplies (\vara - 2 < \varb_1
				\wedge -\vara \geq \varb_3 \wedge \vara \geq \varb_2)}
			\cdot (2\cdot \vara + \varc)}~,
	\]
our algorithm yields (after simplification), writing
\[
 \bxb \eeq \varb_1\geq\varc\ \wedge\ \varb_2<\varb_1+2\ \wedge\ \varb_2\leq-\varb_3~,
\]
the quantifier-free formula: 
\begin{align*}
 &\iverson{\varb_1<\varc}\cdot\infty\\
 &{}+\iverson{\bxb\wedge\varb_1+2\leq-\varb_3\wedge2\cdot\varb_1+\varc+4>0}
       \cdot(2\cdot\varb_1+\varc+4)\\
 &{}+\iverson{\bxb\wedge\varb_1+2>-\varb_3\wedge-2\cdot\varb_3+\varc>0}
       \cdot(-2\cdot\varb_3+\varc)~.
\end{align*}
To see this, proceed by a case distinction: If $\varb_1<\varc$, the implication is true for every $\vara$, and the supremum is $\infty$. Otherwise, $\lqc$ equals $2\cdot\vara+\varc$ on the interval specified by $\varb_2\leq\vara<\varb_1+2$ and $\vara\leq-\varb_3$, and equals $0$ outside that interval. The interval is non-empty exactly when $\bxb$ holds (expressed without referring to $\vara$), in which case its supremal endpoint is $u=\min\{\varb_1+2,-\varb_3\}$, so $\SupOO{\vara}{\lqc}$ evaluates to $\max\{0,2\cdot u+\varc\}$. Splitting on which upper bound supplies $u$, retaining only positive affine values, and omitting zero-valued summands yields the displayed expression. In particular, if $\varb_1\geq\varc$ but the interval specified by $\varb_2\leq\vara<\varb_1+2$ and $\vara\leq-\varb_3$ is empty, $\SupOO{\vara}{\lqc}$ \mbox{evaluates to $0$.}
	\item We provide rigorous soundness proofs, illustrative examples, and upper bounds on the syntactic output size of our algorithm.
	\item\label{contr:craig} We contribute the --- to the best of our knowledge --- first \emph{Craig interpolation theorem} for piecewise linear quantities: Using our quantifier elimination algorithm, we prove that for two \emph{arbitrary} piecewise linear quantities $\lqa,\lqb$ with $\lqa \qentails \lqb $, both the strongest and the weakest 
	$\lqc$ such that 
	\[
		\lqa \qentails \lqc \qentails \lqb 
		\qquad\text{and}\qquad
		\text{the free variables in $\lqc$ are free in \emph{both} $\lqa$ and $\lqb$}
	\]
	are quantifier-free and effectively constructible (see \Cref{ex:craig}).
	\item\label{contr:wp} We apply our quantifier elimination algorithm to \emph{effectively construct closed-form representations} (here, expressions without quantitative quantifiers) for minimal/maximal expected outcomes of loop-free linear\footnote{i.e., all expressions appearing in assignments or guards involve only linear arithmetic. The control flow is unrestricted, i.e., we allow for arbitrarily nested conditionals or probabilistic choices.} probabilistic programs \emph{featuring unbounded non-determinism} (see \Cref{ex:intro:wp}). We also show how our approach enables checking inductivity of candidate upper bounds on minimal/maximal expected outcomes of loops.
\end{enumerate}
%
%
	\begin{example}[Contribution \ref{contr:craig}: Quantitative Craig Interpolation]
		\label{ex:craig}
		Consider: %
		\begin{align*}
			\lqa \eeq& \iverson{\vara \geq 0}\cdot \vara + \iverson{\vara \geq 0 \wedge \varb \leq \vara}\cdot \varb \\
			\lqa' \eeq & \iverson{\vara \geq 0 \wedge \varc \geq \vara} \cdot (2\cdot\vara + \varc +1) + \iverson{\varc < \vara}\cdot \infty
		\end{align*}
		We have $\lqa \qentails \lqa'$. Using our quantifier elimination technique, we effectively construct both the strongest and the weakest\footnote{We say that $\lqc$ is \emph{stronger} (resp.\ \emph{weaker}) than $\lqc'$, if $\lqc\qentails\lqc'$ (resp.\ $\lqc' \qentails \lqc$).} Craig interpolants of \mbox{$(\lqa,\lqa')$ given by}
		\begin{align}
			\underbrace{
				\iverson{\vara \geq 0}\cdot 2\cdot \vara}_{\text{strongest Craig interpolant of $(\lqa,\lqa')$}}
			\quad\text{and}\quad
			\underbrace{
				\iverson{\vara\geq 0}\cdot (3\cdot \vara + 1)}_{\text{weakest Craig interpolant of $(\lqa,\lqa')$}}~.
			\tag*{\finishexample}
		\end{align}
	\end{example}
	
	\begin{remark}\label{rem:simple_interpolants}
		When applying \emph{classical} Craig interpolation for a first-order theory $\theory$ to, e.g., loop invariant generation, \enquote{simple} Craig interpolants, i.e., interpolants that lie strictly \enquote{between} (w.r.t.\ $\models_\theory $) the strongest and the weakest ones, are often very useful \cite{DBLP:conf/cav/AlbarghouthiM13}. Our \emph{quantitative} Craig interpolation technique presented in this paper does \emph{not} aim for obtaining such \enquote{simple} interpolants. Rather, our goal is to prove that quantitative Craig interpolants at all exist and that they are effectively constructible. We discuss possible directions for obtaining simpler quantitative Craig interpolants in \Cref{sec:concl}.
	\end{remark}
	
	\begin{example}[Contribution \ref{contr:wp}: Constructing Minimal/Maximal Expected Outcomes]
		\label{ex:intro:wp}
		Consider the following probabilistic program $\pp$ featuring unbounded non-determinism:
		\begin{align*}
			& \UNDCHOICE{\vara_1}\,; \\
			&\PCHOICE{\ASSIGN{\varb_1}{\varb_1 + \vara_1}}{\nicefrac{1}{2}}{\SKIP}\,; \\
			& \UNDCHOICE{\vara_2}\,; \\
			&\PCHOICE{\ASSIGN{\varb_2}{\varb_2 + \vara_2}}{\nicefrac{1}{2}}{\SKIP} 
		\end{align*}
		Variables range over $\QQ$. $\UNDCHOICE{\vara_1}$ is a purely \emph{non-deterministic} assignment, assigning to $\vara_1$ an \emph{arbitrary} rational number. $\PCHOICE{\ldots}{\nicefrac{1}{2}}{\ldots}$, on the other hand, is a \emph{probabilistic choice}, executing either the left- or the right branch, each with probability $\nicefrac{1}{2}$. The program $\pp$ models a game between two players. Player $1$ (resp.\ Player $2$) starts at some position $\varb_1$ (resp.\ $\varb_2$). Player $1$ then starts by (non-deterministically) choosing some rational $\vara_1$, and then flips a fair coin. If the coin lands heads (say), Player $1$ adds $\vara_1$ to her current position $\varb_1$. Subsequently, it is Player $2$'s turn, who behaves analogously. We now ask:
		\begin{center}
			\emph{Under all possible choices the players can make,} \\
			\emph{what is the minimal (maximal) probability of them ending up at the same position?}
		\end{center}
		Clearly, this (minimal or maximal) probability is \emph{not unique} but \emph{depends} on the initial values of $\varb_1$ and $\varb_2$, and determining these probabilities requires reasoning about (variable valuation-dependent) infima/suprema. By combining our quantifier elimination algorithm with a program logic tailored to reasoning about such probabilities (or more general expected values), we obtain an algorithm that computes the expressions (after simplification)
		\[
		    \textnormal{\textsf{MinProb}}(\varb_1,\varb_2) \eeq \iverson{\varb_1 = \varb_2} \cdot \frac{1}{4}
			\qquad{\text{and}}\qquad
			\textnormal{\textsf{MaxProb}}(\varb_1,\varb_2) \eeq  \frac{3}{4} + \iverson{\varb_1 = \varb_2} \cdot \frac{1}{4}~,
		\]
		mapping initial values of $\varb_1$ and $\varb_2$ to the sought-after minimal/maximal probabilities of the players ending up at the same position. In words: If the players start at the same position $\varb_1 = \varb_2$, then the minimal probability is $\nicefrac{1}{4}$, otherwise $0$. The maximal probability, on the other hand, is $1$ if the players start at the same position, otherwise $\nicefrac{3}{4}$. Put slightly more formally, our algorithm is guaranteed to generate \emph{quantifier-free piecewise linear quantities} serving as closed-form representations for such functions mapping initial variable valuations to \emph{expected outcomes}, i.e., probabilities like the above or more general expected values.
		\hfill $\triangle$
	\end{example}
%
%
%
\paragraph{Related Work.}
Our quantifier elimination algorithm is based on ideas related to Fourier-Motzkin variable elimination \cite{motzkin1936beitraege,fourier1825analyse}. Closely related is symbolic dynamic programming: Sanner, Delgado, and de Barros \cite{DBLP:conf/uai/SannerDB11} treat discrete and continuous state spaces; Zamani, Sanner, and Fang \cite{ZamaniSF12} extend this approach to continuous actions. The latter work introduces symbolic maximization of piecewise defined functions, exploiting disjoint cases (similar to \Cref{Nlemm:FirstLevelSup}) and decomposition into conjunctive guards (similar to \Cref{lemm:ExploitDNF}). Its representation assumes continuity, including at partition boundaries, and finite real-valued expressions; it also permits restricted piecewise quadratic expressions. We instead treat possibly discontinuous, extended-valued piecewise linear quantities and arbitrary finite sequences of supremum and infimum quantifiers placed before the guarded sum. We do not impose continuity, and do not assume that an extremum is attained. Such quantities arise, for example, in quantitative program verification (see \Cref{sec:nondet_programs}). 

Related optimization tasks occur in optimization modulo theories with linear rational costs \cite{SebastianiT15} and symbolic optimization with satisfiability modulo theories (SMT) solvers \cite{LiAKGC14}. These works compute optimal bounds for supplied constraint problems; \cite{LiAKGC14} presents its procedure for a fragment of linear real arithmetic whose constraint sets are closed. In contrast, our elimination task returns a quantity with the uneliminated variables as parameters, and permits alternating quantitative quantifiers. Quantitative quantifier elimination is also related to \emph{parametric programming}, i.e., optimization with some inputs treated as parameters \cite{gstill_parametric}.

Khatami, Pourmahdian, and Tavana \cite{DBLP:journals/fss/TavanaPK24} investigate Craig interpolation for first-order G\"odel logic over relational languages with constants, with truth values in $[0,1]$. Apart from the more restrictive semantic codomain, \cite{DBLP:journals/fss/TavanaPK24} works over arbitrary relational structures with $[0,1]$-valued predicates, whereas we work over the fixed arithmetic structure of the rationals, with $\infty$ and $-\infty$ additionally allowed as quantity values. Baaz and Veith \cite{DBLP:conf/csl/BaazV98} investigate quantifier elimination for first-order logic of fuzzy algebras, whose elements range over the unit interval and whose operations are interpreted by fuzzy truth functions. Teige and Fr\"anzle \cite{TeigeF12} investigate generalized Craig interpolation for stochastic \emph{Boolean} satisfiability. Mahdi and Fr\"anzle \cite{DBLP:conf/rp/MahdiF14} extend this to stochastic satisfiability \emph{modulo theories}, including arithmetic with real-valued free variables; variables in their quantifier prefix have finite discrete domains. Our supremum and infimum quantifiers range over all rational values.

\paragraph{Difference to Conference Paper.} This article extends the conference proceedings version \cite{DBLP:conf/fossacs/BatzKO25} by (i) detailed proofs of the soundness of the quantitative quantifier elimination algorithm, (ii) detailed proofs of the quantitative Craig interpolation theorems, and (iii) Contribution \ref{contr:wp} on effectively constructing minimal/maximal expected outcomes of probabilistic programs with unbounded non-determinism. The proofs corresponding to (i) and (ii) already appeared in the accompanying technical report \cite{BatzKO25report}. Beyond that report, \Cref{sec:nondet_programs} supplies the probabilistic program-analysis application.

\paragraph{Outline.} In \Cref{sec:quantities}, we introduce piecewise linear quantities. We present our quantifier elimination algorithm alongside illustrative examples, essential theorems, and a complexity analysis in \Cref{sec:qelim}. Our quantitative Craig interpolation results are presented in \Cref{sec:craig}. Effectively constructing expected outcomes of non-deterministic probabilistic programs is treated in \Cref{sec:nondet_programs}.
Finally, we conclude in \Cref{sec:concl}.

		\section{Piecewise Linear Quantities}
		\label{sec:quantities}
		 Throughout, we fix a finite non-empty set $\Vars=\{\vara,\varb,\varc,\ldots\}$ of variables. We denote by $\NNO$ the set of natural numbers including $0$ and let $\NN = \NNO \setminus\{0\}$. The set of rationals (resp.\ reals) is denoted by $\QQ$ (resp.\ $\RR$) and we denote by $\ExtQQ = \QQ \cup \{-\infty,\infty\}$ (resp.\ $\ExtRR = \RR \cup \{-\infty,\infty\}$) the set of \emph{extended} rationals (resp.\ reals).
A \emph{(variable) valuation} $\psa\colon\Vars\to\QQ$ assigns a rational number to each variable. The set of all valuations is $\PS$. 

Towards defining piecewise linear quantities and their semantics, we briefly recap linear arithmetic and Boolean expressions. 
\begin{definition}[Linear Arithmetic Expressions]\label{def:LinAX}
The set \LinAX{} of \emph{linear arithmetic expressions} (affine expressions) consists of all expressions \axa{} of the form
\begin{align*}
\axa\quad{}={}\quad \rata_0+\sum\limits_{\idxa=1}^{\cardnum}\rata_\idxa\cdot\vara_\idxa~,
\end{align*}
where $\rata_0,\ldots,\rata_{\cardnum} \in \QQ$ and $\vara_1,\ldots,\vara_{\cardnum}$ enumerate $\Vars$ without repetition.
Moreover, we define the set \ExtLinAX{} of \emph{extended linear arithmetic expressions} as 
\begin{align*}
\ExtLinAX\eeq \LinAX\cup\{-\infty,\infty\}~.
\tag*{\finishdefinition}
\end{align*}
\end{definition}
Notice that every arithmetic expression is normalized in the sense that every variable occurs exactly once. 
We often omit summands $\rata_\idxa \cdot \vara_\idxa$ (resp.\ $\rata_0$) with $\rata_\idxa = 0$ (resp.\ $\rata_0=0$) for the sake of readability. Given $\axa$ as above, \mbox{we denote by}
\[
	\FVars{\axa} \eeq \{ \vara_i \in \Vars ~\mid~ \rata_i \neq 0  \}
\]
the set of (necessarily free) variables occurring in $\axa$. 
 For $\eaxa = \infty$ or $\eaxa = -\infty$, we define $\FVars{\eaxa}  = \emptyset$. Given $\eaxa\in\ExtLinAX$ and $\vara_\idxc\in\Vars$, we say that 
 \begin{align*}
 	\begin{cases}
 		 \vara_\idxc\text{ occurs positively in }\eaxa &  \text{if}~\eaxa = \rata_0+\sum\limits_{\idxa=1}^{\cardnum}\rata_\idxa\cdot\vara_\idxa ~\text{and}~\rata_\idxc>0 \\
 		 \vara_\idxc\text{ occurs negatively in }\eaxa &  \text{if}~\eaxa = \rata_0+\sum\limits_{\idxa=1}^{\cardnum}\rata_\idxa\cdot\vara_\idxa ~\text{and}~\rata_\idxc<0 ~.
 	\end{cases}
 \end{align*}
Finally, given $\psa\in\PS$, the \emph{semantics} $\sem{\psa}{\eaxa} \in \ExtQQ$ of $\eaxa$ under $\psa$ is defined as 
\[
	\sem{\psa}{\eaxa} \eeq 
	\begin{cases}
		\rata_0+\sum\limits_{\idxa=1}^{\cardnum}\rata_\idxa\cdot\psa(\vara_\idxa) &\text{if}~\eaxa = \rata_0+\sum\limits_{\idxa=1}^{\cardnum}\rata_\idxa\cdot\vara_\idxa \\
	-\infty & \text{if}~\eaxa = -\infty \\
	\infty & \text{if}~\eaxa = \infty~.
	\end{cases}
\]

\begin{definition}[Boolean Expressions]\label{def:Bool}
\emph{Boolean expressions} $\bxa$ in the set $\Bool$ adhere to the following grammar, where $\eaxa\in\ExtLinAX$:
\begin{align*}
\bxa \qquad {}\longrightarrow{} \qquad 
 & \eaxa<\eaxa ~\mid~ \eaxa \leq \eaxa ~\mid~ \eaxa > \eaxa ~\mid~ \eaxa \geq \eaxa \tag*{(linear inequalities) \phantom{\finishdefinition}}\\
{}|~& \lnot\bxa\tag*{(negation) \phantom{\finishdefinition}}\\
{}|~& \bxa\land\bxa  \tag*{(conjunction) \phantom{\finishdefinition}}\\
{}|~& \bxa\lor\bxa  \tag*{(disjunction) \finishdefinition}
\end{align*}
%
\end{definition}
A \emph{literal} is an atom of this grammar or its negation. Negated inequalities may be replaced by their complementary inequalities, e.g., $\neg(\eaxa<\eaxb)$ by $\eaxa\geq\eaxb$. A Boolean expression used to select a value or a program branch is called a \emph{guard}.

The Boolean constants and implication are syntactic sugar: $\true$ stands for $0\leq0$, $\false$ for $0<0$, and $\bxa\bimplies\bxb$ for $\neg\bxa\lor\bxb$. Empty conjunctions and disjunctions denote $\true$ and $\false$, respectively. We also write $\eaxa=\eaxb$ for $(\eaxa\leq\eaxb)\land(\eaxb\leq\eaxa)$, and $\eaxa\neq\eaxb$ for $(\eaxa<\eaxb)\lor(\eaxa>\eaxb)$.
We assume that $\lnot$ binds stronger than $\land$ binds stronger than $\lor$, and we use parentheses to resolve ambiguities. The set $\FVars{\bxa}$ of (necessarily free) variables in $\bxa$ is the union of the free-variable sets of its arithmetic operands. A Boolean expression or arithmetic expression is \emph{$\vara$-free} if $\vara$ does not belong to its free-variable set. Given a valuation $\psa$, we write $\psa\models\bxa$ if $\psa$ satisfies $\bxa$ and $\psa\not\models\bxa$ otherwise. We call $\bxa$ \emph{unsatisfiable} if no valuation satisfies it. We write $\bxa\equiv\bxb$ when the same valuations satisfy both expressions. Satisfiability is decidable.\footnote{First simplify comparisons involving $\infty$ or $-\infty$ as described below. The remaining test is SMT solving over LRA; see, e.g., the decision procedure of Dutertre and de Moura \cite{DutertreM06}.}

\paragraph{Normalizing Comparisons with Infinity.}
\label{par:infinity_normalization}
Because valuations assign \emph{finite rational} values to variables, every comparison involving an infinity has a constant truth value, i.e.,
\[
 \text{for all $\axa\in\LinAX$}\colon \quad 
 \axa<\infty~\equiv~\axa\leq\infty~\equiv~\axa>-\infty~\equiv~\axa\geq-\infty~\equiv~\true~,
\]
whereas $\axa\geq\infty$, $\axa>\infty$, $\axa\leq-\infty$, and $\axa<-\infty$ are equivalent to $\false$. Reverse comparisons and comparisons between two infinities are simplified according to the usual extended order on $\ExtRR$. We denote this transformation by $\Norm{\bxa}$. We call its result \emph{infinity-normalized}: all remaining nonconstant comparisons have finite affine operands. We have $\bxa \equiv \Norm{\bxa}$ and $\Norm{\bxa}$ introduces no additional variables.


%

%
\begin{definition}[Piecewise Linear Quantities (adapted from \cite{DBLP:journals/pacmpl/BatzKKM21,DBLP:conf/sas/KatoenMMM10})]\label{def:LQ}
The set $\LQ$ of \emph{(piecewise linear) quantities} consists of all expressions
\begin{align*}
\lqa \quad{}={}\quad \quanta_1\vara_1\ldots\quanta_\nnc\vara_\nnc\colon\SNF~,  
\end{align*}
where $\nnc \in \NNO$, $\nna\in\NN$, and where
\begin{enumerate}
\item\label{def:LQ1} $\quanta_\idxc\in\{\Sup,\Inf\}$ and $\vara_\idxc\in\Vars$ for all $\idxc\in\{1,\ldots,\nnc\}$,
\item\label{def:LQ2} $\bxa_\idxa\in\Bool$ and $\eaxa_\idxa\in\ExtLinAX$ for all $\idxa\in\{1,\ldots,\nna\}$,
\item\label{def:LQ3} for all $\psa\in\PS$ and all $\idxa,\idxc\in\{1,\ldots,\nna\}$ with $\idxa\neq\idxc$, we have\footnote{This is decidable by SMT solving over LRA. Hence, membership in $\LQ$ is decidable.}
 \begin{align*}
 \psa\models\bxa_{\idxa}~\text{and}~\psa\models\bxa_{\idxc}
 \qquad\text{implies}\qquad 
 \eaxa_{\idxa}\neq\infty~\text{or}~\eaxa_{\idxc}\neq-\infty ~.
 \tag*{\finishdefinition}
 \end{align*}
\end{enumerate}
\end{definition}
The leading sequence $\quanta_1\vara_1\ldots\quanta_\nnc\vara_\nnc$ is the \emph{quantifier prefix}; the quantifier-free sum $\sum_{\idxa=1}^{\nna}\iverson{\bxa_\idxa}\cdot\eaxa_\idxa$ following the colon is the \emph{matrix} of $\lqa$. Thus all quantities are in \emph{prenex form}. We call $\bxa_i$ the \emph{guard} and $\eaxa_i$ the \emph{value expression} of the $i$th summand. A summand is \emph{active} at $\psa$ when $\psa\models\bxa_i$. The conditions of \Cref{def:LQ}, in particular the exclusion of simultaneously active opposite infinities, are the \emph{well-formedness} requirements.

Here $\iverson{\bxa}$ is the \emph{Iverson bracket} \cite{iverson_knuth} of the Boolean expression $\bxa$, which evaluates to $1$ under valuation $\psa$ if $\psa\models\bxa$, and to $0$ otherwise. As a quantity, $\iverson{\bxa}$ abbreviates $\iverson{\bxa}\cdot1$, and a constant or affine expression $\eaxa$ abbreviates $\iverson{\true}\cdot\eaxa$.
$\Sup$ is called the \emph{supremum quantifier} and $\Inf$ is the \emph{infimum quantifier}. The quantitative quantifiers take over the role of the classical $\exists$- and $\forall$-quantifiers from first-order predicate logic. Their semantics is detailed further below. 
 If $\nnc=0$, then we call $\lqa$ \emph{quantifier-free}. We often use $\qflqa,\qflqb$ for quantifier-free quantities. Given $\lqa$ as above, the set of \emph{free variables} in $\lqa$ is 
\[
	\FVars{\lqa}\quad{}={}\quad\bigcup\limits_{\idxc=1}^{\nna}\big(\FVars{\bxa_\idxc}\cup\FVars{\eaxa_\idxc}\,\big) \setminus \{\vara_1,\ldots,\vara_\nnc\}~.
\]
A quantity $\lqa$ is likewise called $\vara$-free if $\vara\notin\FVars{\lqa}$. For quantifier-free $\lqa$, we introduce the shorthand $\iverson{\bxa}\cdot\lqa \eeq \sum_{i=1}^n \iverson{\bxa \wedge \bxa_i}\cdot \eaxa_i$.

Towards defining the semantics of quantities, we use the following notions:
Given a valuation $\psa\in\PS$, a variable $\vara\in\Vars$, and $\rata\in\QQ$, we define the valuation obtained from updating the value of $\vara$ under $\psa$ to $\rata$ as
\begin{align*}
\psa\statesubst{\vara}{\rata}(\varb) \quad {}={}\quad
\begin{cases}
\rata & \text{if }\varb=\vara\\
\psa(\varb) & \text{otherwise}~.
\end{cases}
\end{align*}
For distinct variables, $\psa[\vara_1\mapsto\rata_1,\ldots,\vara_j\mapsto\rata_j]$ denotes their simultaneous update, leaving all other variables unchanged. Unless another ordered domain is specified, suprema and infima are taken in $\ExtRR$, with $\sup\emptyset=-\infty$ and $\inf\emptyset=\infty$.

As is standard \cite{aliprantis1981principles} in the realm of extended reals, we define for all $\reala\in\RR$:
\par\noindent
\begin{minipage}[t]{0.48\linewidth}
\begin{enumerate}
	\item $\reala+\infty = \infty+\reala= \infty$
	\item $\reala+(-\infty) = -\infty+\reala =  -\infty$
	\item $\infty+\infty = \infty$ 
	\item $-\infty+(-\infty) = -\infty$ 
	\item $-\infty \cdot 0 = 0 \cdot (-\infty) = 0 = 0\cdot \infty = \infty \cdot 0$
\end{enumerate}
\end{minipage}\hfill
\begin{minipage}[t]{0.48\linewidth}
\begin{enumerate}[start=6]
	\item if $\reala >0$, then $\reala\cdot\infty = \infty \cdot \reala = \infty$
	\item if $\reala >0$, then $\reala\cdot(-\infty) = -\infty \cdot \reala = -\infty$
	\item if $\reala <0$, then $\reala\cdot\infty = \infty \cdot \reala = -\infty$
	\item if $\reala <0$, then $\reala\cdot(-\infty) = -\infty \cdot \reala = \infty$
\end{enumerate}
\end{minipage}
\par\medskip
The terms $\infty+(-\infty)$ and $-\infty+\infty$ are undefined. The condition from \Cref{def:LQ}.\ref{def:LQ3} ensures that we never encounter such undefined terms, which yields the semantics of piecewise linear quantities to be well-defined:
%
%
\begin{definition}[Semantics of Piecewise Linear Quantities]\label{def:SemOfLQ}
Let $\lqa\in\LQ$ and $\psa\in\PS$. The \emph{semantics\footnote{It follows from the soundness of our quantifier elimination algorithm (\Cref{thm:algo_sound}) that all $\lqa\in\LQ$ evaluate to extended rationals in $\ExtQQ$.}} $\sem{\psa}{\lqa}\in\ExtRR$ of $\lqa$ under $\psa$ is defined inductively: 
\begin{align*}
\sem{\psa}{\SNF} & \quad{}={}\quad\sum\limits_{\idxa=1}^n 
\begin{cases}
	\sem{\psa}{\eaxa_\idxa} &\text{if}~\psa \models \bxa_i \\
	0 &\text{if}~\psa \not\models \bxa_i
\end{cases}
	\\
\sem{\psa}{\SupOO{\vara}{\lqb}}&\quad{}={}\quad\sup\Bigl\{\sem{\psa\statesubst{\vara}{\rata}}{\lqb}\mid\rata\in\QQ\Bigr\}\\
\sem{\psa}{\InfOO{\vara}{\lqb}}&\quad{}={}\quad\inf\Bigl\{\sem{\psa\statesubst{\vara}{\rata}}{\lqb}\mid\rata\in\QQ\Bigr\}~.
\tag*{\finishdefinition}
\end{align*} 
%
\end{definition}
In words, if $\lqa$ is quantifier-free, then $\sem{\psa}{\lqa}$ evaluates to the sum of all extended arithmetic expressions $\eaxa_\idxc$ for which $\psa \models \bxa_\idxc$. The semantics of $\Sup$ and $\Inf$ makes it evident that the quantitative quantifiers generalize the classical quantifiers: Whereas $\exists$ maximizes --- so to speak --- a truth value, the $\Sup$-quantifier maximizes a quantity by evaluating to the supremum obtained from evaluating $\lqa$ under all possible values for $\vara$. The semantics of $\Inf$ behaves analogously by evaluating to an infimum.

Finally, we say that two piecewise linear quantities $\lqa,\lqb\in\LQ$ are \emph{(semantically) equivalent}, denoted by $\lqa\equiv\lqb$, if $\sem{\psa}{\lqa}=\sem{\psa}{\lqb}$ for all $\psa \in \PS$.
\label{def:SE}

		\section{Quantitative Quantifier Elimination}
		\label{sec:qelim}
In this section, we detail our quantifier elimination procedure alongside illustrative examples. Given an \emph{arbitrary} piecewise linear quantity
\[
		\lqa \quad{}={}\quad \quanta_1\vara_1\ldots\quanta_\nnc\vara_\nnc\colon\SNF~{}\in{}~\LQ~,  
\]
we aim to automatically compute some $\QE{\lqa} \in \LQ$ satisfying
\[
	\QE{\lqa}~\text{is quantifier-free}\qquad\text{and}\qquad\QE{\lqa}~\text{is equivalent to $\lqa$, i.e., $\lqa \equiv \QE{\lqa}$}~.
\]
As with classical quantifier elimination, it suffices being able to eliminate piecewise linear quantities containing \emph{a single} quantifier, which then enables to process quantities containing an \emph{arbitrary} number of quantifiers in an inner- to outermost fashion, i.e.,
\[
	\QE{\lqa} \eeq \QE{\quanta_1\vara_1 \colon \QE{\ldots\QE{\quanta_\nnc\vara_\nnc\colon\SNF}}}~.
\] 
%
%
Throughout the next sections, we thus fix an $\lqa$ of the form 
\begin{align}
	\label{eqn:f_stage0}
	\lqa \quad{}={}\quad\quanta\vara \colon\SNF~,
\end{align}
where $\quanta \in \{\Sup,\Inf\}$. 
We proceed by means of a $3$-level divide-and-conquer approach. We describe each of the involved stages in Sections \ref{sec:stage1}-\ref{sec:stage3}. In \Cref{sec:algo}, we summarize our approach by providing an algorithm.
%
%
%
\subsection{Stage 1: Exploiting the Guarded Normal Form}
\label{sec:stage1}
First, we transform the input $\lqa$ into a normal form (extending \cite[Section 5.1]{DBLP:conf/sas/KatoenMMM10}), which enables us to subdivide the quantifier elimination problem into simpler sub-problems. This normal form enforces a more explicit representation of the $\ExtRR$-valued function a piecewise linear quantity represents.
\begin{definition}[Guarded Normal Form]
	\label{def:GNF}
	Let $\lqc \in \LQ$ be given by 
	\[
	\lqc \quad{}={}\quad \quanta_1\vara_1\ldots\quanta_\nnc\vara_\nnc\colon\SNF
	\]
	and fix some variable $\vara \in \Vars$. 
	We say that $\lqc$ is in \emph{guarded normal form w.r.t.\ $\vara$}, denoted by $\lqc \in \XGNF{\vara}$, if all of the following conditions hold:
	\begin{enumerate}
		\item\label{def:GNF1} the $\bxa_i$ partition the set $\PS$ of valuations, i.e., for all $\psa\in\PS$ there exists exactly one $\idxa\in\{1,\ldots,\nna\}$ such that $\psa\models\bxa_\idxa$,
		\item\label{def:GNF2} the $\bxa_\idxa$ are in disjunctive normal form (\DNF), i.e., 
		\[
		\forall i \in \{1,\ldots,n\}\colon \quad \bxa_\idxa~\text{is of the form}~\bigvee\limits_{\idxc}\bigwedge\limits_{\idxd}\lit_{\idxc,\idxd}~,
		\]
		where each $\lit_{\idxc,\idxd} \in \Bool$
		is a (strict or non-strict) inequality between finite expressions in $\LinAX$, after infinity normalization. Constant guards are represented by $0\leq0$ or $0<0$,
		\item\label{def:GNF3} each linear inequality $\lit$ in $\lqc$ is either $\vara$-free or of the form
		\[
		\lit \quad{}={}\quad \vara\sim\eaxb
		\]
		for some $\eaxb\in\LinAX$ with $\vara\not\in\FVars{\eaxb}$ and $\sim\hspace{1mm}\in\{>,\geq,<,\leq\}$.   
		\hfill \finishdefinition
	\end{enumerate}
\end{definition}
If Condition \ref{def:GNF}.\ref{def:GNF1} holds, then we say that $\lqc$ is \emph{partitioning}. An inequality containing $\vara$ in the form required by Condition \ref{def:GNF}.\ref{def:GNF3} is called \emph{$\vara$-isolated}.

The following construction is adapted from \cite{DBLP:conf/sas/KatoenMMM10}:
\begin{definition}[Partitioning]\label{def:partition}
Let $\lqc=\quanta_1\vara_1\ldots\quanta_\nnc\vara_\nnc\colon\SNF$. For $J\subseteq\{1,\ldots,n\}$, let
\[
 \bxb_J\eeq\bigwedge_{i\in J}\bxa_i\ \wedge\ \bigwedge_{i\in\{1,\ldots,n\}\setminus J}\neg\bxa_i~.
\]
We define $\PART{\lqc}$ by retaining the quantifier prefix of $\lqc$ and replacing its matrix by
\[
 \sum_{\substack{J\subseteq\{1,\ldots,n\}\\\bxb_J\text{ satisfiable}}}
       \iverson{\bxb_J}\cdot\left(\sum_{i\in J}\eaxa_i\right)~.
\]
The empty sum is identified with $0$. The \emph{cell specified by $\bxb_J$} is the set $\{\psa\in\PS\mid\psa\models\bxb_J\}$.  \hfill\finishdefinition
\end{definition}
$\PART{\lqc}$ is partitioning and equivalent to $\lqc$. On a satisfiable cell, \Cref{def:LQ}.\ref{def:LQ3} excludes opposite infinities among the active values, so the displayed sum is well-defined. There is at least one satisfiable cell because $\PS$ is non-empty.

If $\lqc$ has $n$ summands and every original guard contains at most $m$ inequalities, then $\PART{\lqc}$ has at most $2^n$ summands and at most $n\cdot m$ inequalities per guard. 

Speaking of a \emph{normal form} in \Cref{def:GNF} is therefore justified: after partitioning, we normalize all infinity comparisons, transform guards into DNF, and normalize each finite inequality with respect to $\vara$. Before isolating $\vara$, we collect its coefficients across both sides of the comparison. If the resulting coefficient is zero, we replace the comparison by the resulting $\vara$-free comparison. Otherwise, we isolate $\vara$, reversing the inequality when dividing by a negative coefficient. These operations preserve equivalence and introduce \mbox{no free variables.}

\begin{example}
	\label{ex:running_ex_1}
	Recall the piecewise linear quantity from \Cref{ex:intro} given by
	\begin{align*}
		 \Sup \vara \colon \iverson{\varb_1 \geq \varc \bimplies (\vara - 2 < \varb_1
		 	\wedge -\vara \geq \varb_3 \wedge \vara \geq \varb_2)}\cdot (2\cdot \vara + \varc) ~,
	\end{align*}
	which is \emph{not} in $\XGNF{\vara}$. Applying the partitioning construction from above yields
	\begin{align*}
		 &\Sup \vara \colon \iverson{\varb_1 < \varc \vee (\vara  < \varb_1 + 2 
		 	\wedge \vara \leq -\varb_3\wedge \vara \geq \varb_2)}\cdot (2\cdot \vara + \varc) \\
		 &\qquad {}+{} \iverson{(\varb_1 \geq \varc \wedge \vara  \geq \varb_1 + 2) 
		 	\vee (\varb_1 \geq \varc \wedge \vara > -\varb_3)\vee ( \varb_1 \geq \varc \wedge \vara < \varb_2)}\cdot 0~.
	\end{align*}
	which \emph{is} in $\XGNF{\vara}$ and will serve us as a running example.
	\hfill \finishexample
\end{example}
Now assume w.l.o.g.\ that the input quantity  $\lqa$ is in $\XGNF{\vara}$.
Each of the Conditions \ref{def:GNF}.\ref{def:GNF1}-\ref{def:GNF3} is essential to our approach. We will now exploit that $\lqa$ is partitioning. Given $\bxa \in \Bool$ and $\eaxa \in \ExtLinAX$, we define the shorthands
\begin{align*}
	\bxa\searrow\eaxa\eeq\down{\bxa}{\eaxa}
	\qquad\text{and}\qquad
	\bxa\nearrow\eaxa\eeq\up{\bxa}{\eaxa}~.
\end{align*}
Notice that these quantities are always partitioning.
Now consider the following:
\begin{restatable}{theorem}{NlemFirstLevSup}
	\label{Nlemm:FirstLevelSup}
	Let $\vara\in\Vars$ and let $\SNF\in\XGNF{\vara}$. We have for all $\psa\in\PS$:
	\begin{enumerate}
		\item $\sem{\psa}{\SupOO{\vara}{\SNF}}=\max\bigl\{\sem{\psa}{\SupOO{\vara}{(\bxa_\idxa\searrow\eaxa_\idxa)}}\mid\idxa\in\{1,\ldots,\nna\}\bigr\}$
		\item $\sem{\psa}{\InfOO{\vara}{\SNF}}=\min\bigl\{\sem{\psa}{\InfOO{\vara}{(\bxa_\idxa\nearrow\eaxa_\idxa)}}\mid\idxa\in\{1,\ldots,\nna\}\bigr\}$
	\end{enumerate}
\end{restatable}
\begin{proof}
	This is a consequence of the fact that the quantity is partitioning and that $-\infty$ (resp.\ $\infty$) are neutral wr.t. $\max$ (resp.\ $\min$).
	See \Cref{Nproof:lemm:FirstLevelSup} \mbox{for details.}
\end{proof}
We may thus consider each summand of the input quantity $\lqa$ separately. Assuming we can compute $\QE{\SupOO{\vara}{\bxa_\idxa\searrow\eaxa_\idxa}}$ and $\QE{\InfOO{\vara}{\bxa_\idxa\nearrow\eaxa_\idxa}}$, we obtain the 
sought-after quantifier-free equivalent of $\lqa$ by effectively constructing valuation-wise minima and maxima of finite sets of partitioning quantities as follows:
\begin{lemma}
	\label{lemm:statewise_maxmin}
	Let $\seta=\{\qflqa_1,\ldots,\qflqa_\nna\} \subseteq \LQ$ for some $\nna\geq 1$, where each 
	\[ 
		\qflqa_\idxa 
		\quad{}={}\quad
		\sum\limits_{\idxc=1}^{\nnb_\idxa}[\bxa_{\idxa,\idxc}]\cdot\eaxa_{\idxa,\idxc}
	\]
	is quantifier-free and partitioning. Then:
	\begin{enumerate}
		\item There is an effectively constructible $\MAX(\seta) \in \LQ$ such that
		\[
		\forall \psa\in\PS \colon \quad \sem{\psa}{\MAX(\seta)} \eeq \max\bigl\{\sem{\psa}{\qflqa_\idxa}\mid\idxa\in\{1,\ldots,\nna\}\bigr\}~.
		\]
		\item There is an effectively constructible $\MIN(\seta) \in \LQ$ such that 
		\[
		\forall \psa\in\PS \colon \quad \sem{\psa}{\MIN(\seta)} \eeq \min\bigl\{\sem{\psa}{\qflqa_\idxa}\mid\idxa\in\{1,\ldots,\nna\}\bigr\}~.
		\]
	\end{enumerate}
	Moreover, both $\MAX(\seta)$ and $\MIN(\seta)$ are partitioning.
\end{lemma}
\begin{proof}
	Write $\underline{\nnb}_\idxa = \{1,\ldots, \nnb_\idxa\}$. We construct\footnote{As usual, the empty conjunction is equivalent to $\true$.}
	\begin{align*}
		&\MAX(\seta)  = \sum\limits_{(\idxc_1,\ldots,\idxc_\nna)\in \underline{\nnb}_1\times\ldots\times\underline{\nnb}_\nna}\quad\sum\limits_{\idxa=1}^{\nna}\\
		&\qquad \qquad\qquad\Bigl[
		\underbrace{\bigwedge\limits_{\nnc=1}^{\nna} \bxa_{\nnc,\idxc_\nnc}}_{
			\text{$\qflqa_\nnc$ evaluates to $\eaxa_{\nnc,\idxc_\nnc}$}
		}
		\land
		\underbrace{\bigwedge\limits_{\nnc= 1}^{\idxa - 1} \eaxa_{\idxa,\idxc_\idxa}>\eaxa_{\nnc,\idxc_\nnc}
			\land
			\bigwedge\limits_{\nnc= \idxa + 1}^{\nna} \eaxa_{\idxa,\idxc_\idxa}\geq \eaxa_{\nnc,\idxc_\nnc}}_{\substack{\text{$\qflqa_\idxa$ is the quantity with smallest index} \\ \text{evaluating to the sought-after maximum}}}
		\Bigr] \cdot \eaxa_{\idxa,\idxc_\idxa}~.
	\end{align*}
	$\MAX(\seta)$ iterates over all combinations of summands, which determine the value each of the $\qflqa_\idxa$ evaluate to (first summand). For each of these combinations, we check, for each $\idxa \in \{1,\ldots,\nna\},$ whether $\qflqa_\idxa$ evaluates to the sought-after maximum (second summand). $\MAX(\seta)$ is partitioning since the $\qflqa_\idxa$ are partitioning and $\MAX(\seta)$ selects the maximizing quantity with the \emph{smallest index}. The construction of $\MIN(\seta)$ is analogous (see \Cref{app:minima}). 
\end{proof}
Combining \Cref{Nlemm:FirstLevelSup} and \Cref{lemm:statewise_maxmin} thus yields:
\begin{enumerate}
	\item $\QE{\SupOO{\vara}{\SNF}}  \eeq \MAX\left(\bigl\{\QE{\SupOO{\vara}{(\bxa_\idxa\searrow\eaxa_\idxa)}}\mid\idxa\in\{1,\ldots,\nna\}\bigr\} \right)$ 
	\item $\QE{\InfOO{\vara}{\SNF}}  \eeq \MIN\left(\bigl\{\QE{\InfOO{\vara}{(\bxa_\idxa\nearrow\eaxa_\idxa)}}\mid\idxa\in\{1,\ldots,\nna\}\bigr\}\right)$
\end{enumerate}
\begin{example}
	\label{ex:running_ex_2}
	Continuing \Cref{ex:running_ex_1}, we have for every $\psa\in\PS$,
	\begin{align*}
		 \QE{\lqa} 
		%
		\eeq& 
		 \MAX \big(\bigl\{
		\QE{\Sup \vara \colon \big(\varb_1 < \varc \vee (\vara  < \varb_1 + 2 \wedge \vara \leq -\varb_3 \wedge \vara \geq \varb_2)\big) \searrow 2\cdot \vara + \varc}, \\
		&\qquad \qquad 
		\QE{\Sup \vara \colon \big((\varb_1 \geq \varc \wedge \vara  \geq \varb_1 + 2)
			\vee \ldots)
			 \searrow 0}
		\bigr\}\big)~.
		\tag*{$\finishexample$}
\end{align*}
	%
\end{example}

\subsection{Stage 2: Exploiting the Disjunctive Normal Form}
\label{sec:stage2}
In this stage, we aim to eliminate the quantifiers from the simpler quantities
\[
	\SupOO{\vara}{(\bxa\searrow\eaxa)}
	\qquad\text{or}\qquad 
	\InfOO{\vara}{(\bxa\nearrow\eaxa)}~.
\]
Recall that we assume the input quantity $\lqa$ to be in guarded normal form w.r.t.~$\vara$, which yields the Boolean expression $\bxa$ to be in DNF (cf.\ \Cref{def:GNF}.\ref{def:GNF2}). Exploiting the disjunctive shape of $\bxa$ yields the following:
\begin{restatable}{theorem}{lemmExploitDNF}
	\label{lemm:ExploitDNF}
	Let $\eaxa \in\ExtLinAX$ be an extended arithmetic expression and let 
	\[
		\bxa
		\quad{}={}\quad
		\bigvee\limits_{\idxa=1}^n\disj_\idxa ~{}\in{}~ \Bool
	\]
	be a Boolean expression in DNF for some $n\geq 1$. We have for all $\psa\in\PS$:
	\begin{enumerate}
		\item $\sem{\psa}{\SupOO{\vara}{(\bxa\searrow\eaxa)}}=\max\bigr\{\sem{\psa}{\SupOO{\vara}{(\disj_\idxa\searrow\eaxa)}}\mid\idxa\in\{1,\ldots,\nna\}\bigr\}$
		\item $\sem{\psa}{\InfOO{\vara}{(\bxa\nearrow\eaxa)}}=\min\bigr\{\sem{\psa}{\InfOO{\vara}{(\disj_\idxa\nearrow\eaxa)}}\mid\idxa\in\{1,\ldots,\nna\}\bigr\}$
	\end{enumerate}
\end{restatable}
\begin{proof}
	We first observe that 
	\begin{align*}
		 \sem{\psa}{\SupOO{\vara}{(\bxa\searrow\eaxa)}} \eeq \sup \Bigl(\bigcup\limits_{\idxa=1}^n \Bigl\{\sem{\psa\statesubst{\vara}{\rata}}{[\disj_\idxa]\cdot\eaxa}\mid\rata\in\QQ\text{ and }\psa[\vara\mapsto\rata]\models\disj_\idxa\Bigr\}\Bigr)
	\end{align*}
	and then make use of the fact that the supremum of a finite union of extended reals is the maximum of the individual suprema, i.e., the above is equal to
	\begin{align*}
		&\max\Bigl(\bigcup\limits_{\idxa=1}^n \Bigl\{\sup\Bigl\{\sem{\psa\statesubst{\vara}{\rata}}{\disj_\idxa\searrow\eaxa}\mid\rata\in\QQ\Bigr\}\Bigr\}\Bigr)
		\tag{$-\infty$ is neutral w.r.t.\ $\sup$}\\
		\eeq&\max\Bigl(\bigcup\limits_{\idxa=1}^ n\bigl\{\sem{\psa}{\SupOO{\vara}{(\disj_\idxa\searrow\eaxa)}}\bigr\}\Bigr)\tag{\Cref{def:SemOfLQ}}\\
		\eeq&\max\bigr\{\sem{\psa}{\SupOO{\vara}{(\disj_\idxa\searrow\eaxa)}}\mid\idxa\in\{1,\ldots,\nna\}\bigr\}
		\tag{rewrite set}~.
	\end{align*}
	The reasoning for $\Inf$ is analogous.
	See \Cref{proof:lemm:ExploitDNF} for a detailed proof.
\end{proof}
Hence, combining \Cref{lemm:ExploitDNF} with \Cref{lemm:statewise_maxmin} reduces our problem further to eliminating quantifiers from the above simpler quantities. Put formally:
\begin{enumerate}
	\item $\QE{\SupOO{\vara}{(\bxa\searrow\eaxa)}} \eeq \MAX\big(\bigr\{{\QE{\SupOO{\vara}{(\disj_\idxa\searrow\eaxa)}}}\mid\idxa\in\{1,\ldots,\nna\}\bigr\} \big)$ \\
	\item $\QE{\InfOO{\vara}{(\bxa\nearrow\eaxa)}} \eeq \MIN\big(\bigr\{{\QE{\InfOO{\vara}{(\disj_\idxa\nearrow\eaxa)}}}\mid\idxa\in\{1,\ldots,\nna\}\bigr\}\big)$
\end{enumerate}
\begin{example}
	\label{ex:running_ex_3}
	Continuing \Cref{ex:running_ex_2}, we have 
	\begin{align*}
	&\QE{\Sup \vara \colon \big(\varb_1 < \varc \vee (\vara  < \varb_1 + 2 \wedge \vara \leq -\varb_3 \wedge \vara \geq \varb_2)\big) \searrow 2\cdot \vara + \varc} \\
	\eeq&\MAX\big(\big\{
		\QE{\Sup \vara \colon \varb_1 < \varc \searrow 2\cdot \vara + \varc}, \\
		&\qquad\qquad \QE{\Sup \vara \colon (\vara  < \varb_1 + 2 \wedge \vara \leq -\varb_3 \wedge \vara \geq \varb_2) \searrow 2\cdot \vara + \varc}
	 \big\}\big) ~.
	\end{align*}
	The second argument of $\MAX$ from \Cref{ex:running_ex_2} is treated analogously.
	\hfill\finishexample
	%
\end{example}

\subsection{Stage 3: Computing Valuation-Dependent Suprema and Infima}
\label{sec:stage3}
This is the most involved stage since we need to operate on the atomic level of the given expressions. We aim to eliminate the quantifiers from \mbox{quantities of the form}
\[
\SupOO{\vara}{(\bigwedge_{i=1}^n \lit_i\searrow\eaxa)}
\qquad\text{or}\qquad 
\InfOO{\vara}{(\bigwedge_{i=1}^n \lit_i \nearrow\eaxa)}~,
\]
where each $\lit_i$ is a linear inequality. We start with an example.
\begin{example}
	\label{ex:running_ex_4}
	Continuing \Cref{ex:running_ex_3}, we perform quantifier elimination on
	\[
	\lqc \eeq \Sup \vara \colon \underbrace{(\vara  < \varb_1 + 2 \wedge \vara \leq -\varb_3 \wedge \vara \geq \varb_2)}_{{}=\,\disj} \searrow \underbrace{\vphantom{(}2\cdot \vara + \varc}_{{}=\,\eaxa}~.
	\]
	Fix some valuation $\psa$. First observe that if there is no $\rata \in \QQ$ such that $\psa\statesubst{\vara}{\rata} \models \disj$ --- or, phrased in predicate logic, if $\psa \not\models \exists \vara \colon \disj$ ---,
	then $\sem{\psa}{\lqc}$ evaluates to $-\infty$. Otherwise, we need to inspect $\disj$ and $\eaxa$ closer in order to determine $\sem{\psa}{\lqc}$. Hence, eliminating the $\Sup \vara$-quantifier involves characterizing whether $\psa \models \exists \vara \colon \disj$ holds \emph{without referring to $\vara$}. This boils down to performing \emph{classical} quantifier elimination on the formula $\exists \vara \colon \disj$. We leverage classical Fourier-Motzkin variable elimination: Compare the bounds $\disj$ imposes on $\vara$ and encode whether they are consistent. Towards this end, we construct
	\[
		\phiSAT(\disj,\vara) 
		\quad{}={}\quad
		\underbrace{\varb_2 < \varb_1 +2 \wedge \varb_2 \leq -\varb_3}_{\text{equivalent to $\exists \vara \colon \disj$}}
			\quad{}\in\quad\Bool~.
		%
	\]
	Now, how can we characterize $\sem{\psa}{\lqc}$ in case $\psa \models \phiSAT(\disj,\vara)$? 
	We first observe that $\vara$ occurs positively in $\eaxa$. Therefore, intuitively, the $\Sup \vara$-quantifier aims to maximize the value of $\vara$ under all possible assignments satisfying $\disj$. Since we isolate $\vara$ in every inequality where $\vara$ occurs, we can readily read off $\disj$ that $\vara$'s maximal (or, in fact, \emph{supremal}) value is given by the \emph{minimum} of $\psa(\varb_1) + 2$ and $-\psa(\varb_3)$ --- the least of all upper bounds imposed on $\vara$. Overall, we get
	\begin{align*}
		\lqc ~{}\equiv{}~ &
		\iverson{\phiSAT(\disj,\vara)} \cdot \big(\iverson{\varb_1+2 \leq -\varb_3}\cdot (2\cdot \varb_1 + \varc +4) \\
		&\quad {}+\iverson{\varb_1+2 > -\varb_3}\cdot (-2\cdot \varb_3 + \varc) \big) {}+ \iverson{\neg\phiSAT(\disj,\vara)}\cdot (-\infty)
		%
	\end{align*}
	The above quantifier-free equivalent of $\lqc$ indeed evaluates to $-\infty$ if $\psa\not\models\phiSAT(\disj,\vara)$ and, otherwise, performs a case distinction on said least upper bounds on $\vara$. 
	
	Finally, consider the other quantity of \Cref{ex:running_ex_3}:
	%
	\[
	\lqc' \eeq \Sup \vara \colon \underbrace{\varb_1 < \varc}_{{}=\,\disj'} \searrow \underbrace{2\cdot \vara + \varc}_{{}=\,\eaxa'}
	\]
	and observe that $\disj'$ does not impose any bound on $\vara$ whatsoever. This highlights the need for a careful treatment of $\infty$ (or, in dual situations, $-\infty$): Since $\vara$ occurs positively in $\eaxa'$, we have $\sem{\psa}{\lqc'} = \infty$ whenever $\psa \models \varb_1 < \varc$, and $\sem{\psa}{\lqc'} = -\infty$ otherwise. We thus have 
	\begin{align*}
		\lqc' ~{}\equiv{}~ \iverson{\varb_1 < \varc}\cdot \infty + \iverson{\varb_1 \geq \varc} \cdot (-\infty)~.
	\end{align*}
	When considering $\Inf \vara$-quantifiers or when $\vara$ occurs negatively in $\eaxa$, the above described observations need to be dualized, which we detail further below.
	\hfill \finishexample
\end{example}
We condense the following steps for eliminating the $\Sup \vara$- or $\Inf \vara$-quantifiers:
\begin{enumerate}
	\item Extract lower and upper bounds on $\vara$ imposed by the \mbox{Boolean expression $\disj$.}
	\item Construct the Boolean expression $\phiSAT(\disj,\vara)$ via classical Fourier-Motzkin.
	\item Characterize least upper- and greatest lower bounds on $\vara$ admitted by $\disj$.
	\item Eliminate the $\Sup \vara$- or $\Inf \vara$-quantifiers by gluing the above concepts together.
\end{enumerate}
We detail these steps in the subsequent paragraphs. Fix $\vara \in \Vars$, $n\geq1$, and an infinity-normalized conjunction with $\vara$ isolated as in \Cref{def:GNF}.\ref{def:GNF3}: 
\[
	\disj \quad{}={}\quad \bigwedge_{i=1}^n \lit_i~.
\]
For $\psa\in\PS$, define the \emph{fibre}
\[
 \Feas{\psa}{\disj,\vara}\eeq
 \{\rata\in\QQ\mid\psa\statesubst{\vara}{\rata}\models\disj\}~,
\]
which is the set of allowed values of $\vara$ in $\disj$ after fixing all other variables as in $\psa$. We call this fibre \emph{feasible} when it is non-empty.

\paragraph{Extracting Lower and Upper Bounds.}
Given $\sim\hspace{1mm}\in\bounds$, we define
	%
	\begin{align*}
		&\Bounds{\sim}{\vara}{\disj}\\
		\eeq &\begin{cases}
			\{\eaxa\in \ExtLinAX\mid\exists i\in \{1,\ldots,n\} \colon\lit_i =\vara\sim\eaxa\} & \text{if}~ \sim\hspace{1mm}\in\{>,<\}\\
			\{\eaxa\in \ExtLinAX\mid\exists i\in \{1,\ldots,n\} \colon\lit_i=\vara\sim\eaxa\}\cup\{-\infty\} & \text{if}~\sim\hspace{1mm}=\hspace{1mm}\geq\\
			\{\eaxa\in \ExtLinAX\mid\exists i\in \{1,\ldots,n\} \colon\lit_i=\vara\sim\eaxa\}\cup\{\infty\} &\text{if}~\sim\hspace{1mm}=\hspace{1mm}\leq
		\end{cases}
	\end{align*}  
	%
and let $\UBounds{\vara}= \Bounds{<}{\vara}{\disj}\cup\Bounds{\leq}{\vara}{\disj} $ and $\LBounds{\vara} = \Bounds{>}{\vara}{\disj}\cup\Bounds{\geq}{\vara}{\disj}$.
Including $\infty$ and $-\infty$, respectively, by default will be convenient when characterizing least upper- and greatest lower bounds on $\vara$ admitted by $\disj$: If there is no upper (resp.\ lower) bound on $\vara$ whatsoever imposed by $\disj$, our construction will automatically default to $\infty$ (resp.\ $-\infty$). 
\paragraph{Classical Fourier-Motzkin Quantifier Elimination with Infinity.}
	We define
	\begin{align*}
		\phiSAT(\disj,\vara) \eeq
		&\bigwedge\limits_{\substack{\eaxa\in\Bounds{\geq}{\vara}{\disj},\\\eaxb\in\Bounds{\leq}{\vara}{\disj}}}\eaxa\leq\eaxb
		\land
		\bigwedge\limits_{\substack{\eaxa\in\Bounds{\geq}{\vara}{\disj},\\\eaxb\in\Bounds{<}{\vara}{\disj}}}\eaxa<\eaxb
		\land
		\bigwedge\limits_{\substack{\eaxa\in\Bounds{>}{\vara}{\disj},\\\eaxb\in\Bounds{\leq}{\vara}{\disj}}}\eaxa<\eaxb
		\\
		&\quad{}\land\bigwedge\limits_{\substack{\eaxa\in\Bounds{>}{\vara}{\disj},\\\eaxb\in\Bounds{<}{\vara}{\disj}}}\eaxa<\eaxb
		\land
		\bigwedge\limits_{\substack{i \in \{1,\ldots,n\},\\\vara\not\in\FVars{\lit_i}}}\lit_i ~.
	\end{align*}
%
as is standard in Fourier-Motzkin variable elimination. Under the infinity-normalization invariant, the only infinite bounds are the default lower bound $-\infty$ and upper bound $\infty$. This restriction is essential: without normalization, $\disj=(\vara\geq\infty)$ would incorrectly pass the weak pairwise comparison with the default upper bound.
\begin{lemma}\label{lemm:phiSAT}
Let $\disj$ be an infinity-normalized conjunction of finite linear inequalities, with $\vara$ isolated as in \Cref{def:GNF}.\ref{def:GNF3}. Then $\vara\notin\FVars{\phiSAT(\disj,\vara)}$ and, for all $\psa\in\PS$,
\[
 \psa\models\phiSAT(\disj,\vara)
 \qquad\text{iff}\qquad
 \{\rata\in\QQ\mid\psa[\vara\mapsto\rata]\models\disj\}\neq\emptyset~.
\]
\end{lemma}
\begin{proof}
 $\phiSAT(\disj,\vara)$ is the classical Fourier-Motzkin projection, with additional tautological comparisons involving the default bounds \cite{fourier1825analyse,motzkin1936beitraege}. The equivalence follows from Fourier-Motzkin correctness, and \mbox{$\vara$-freeness follows from isolation.}
\end{proof}

\paragraph{Characterizing Suprema and Infima.}

	\label{def:phiInf+phiSup}
	Fix an arbitrary enumeration of the bounds on $\vara$, i.e.,
	 let $\UBounds{\vara} = \{\eaxu_1,\ldots,\eaxu_{m_1}\}$ and $\LBounds{\vara}= \{\eaxl_1,\ldots,\eaxl_{m_2}\}$. \mbox{We define:}
	\begin{enumerate}
		\item $\phiSup(\disj,\vara,\eaxu_i) = \bigwedge_{k=1}^{i-1} \eaxu_i < \eaxu_k \wedge \bigwedge_{k=i+1}^{m_1} \eaxu_i \leq \eaxu_k$
		\item $\phiInf(\disj,\vara,\eaxl_i) = \bigwedge_{k=1}^{i-1} \eaxl_i > \eaxl_k \wedge \bigwedge_{k=i+1}^{m_2} \eaxl_i \geq \eaxl_k$ 
		%
	\end{enumerate}
	%
	%
%
%
When the fibre $\Feas{\psa}{\disj,\vara}$ is feasible, $\phiSup(\disj,\vara,\eaxu_i)$ evaluates to $\true$ under valuation $\psa$ if $\sem{\psa}{\eaxu_i}$ evaluates to the \emph{least upper bound} on $\vara$ admitted by $\disj$ under $\psa$ with \emph{minimal} $i$. %
The intuition for $\phiInf(\disj,\vara,\eaxl_i)$ is analogous. Put more formally:
\begin{restatable}{theorem}{thmPhiInfSup}
	\label{lemm:phiInf+phiSup}
	Let $\disj$ satisfy the hypotheses of \Cref{lemm:phiSAT}, and let $\psa\in\PS$ such that  
	\[
		\bigl\{\rata\in\QQ\mid\psa[\vara\mapsto\rata]\models\disj\bigr\} ~{}\neq{}~ \emptyset~.
	\]
	Then all of the following statements hold:
	\begin{enumerate}
		\item\label{lemm:phiInf+phiSup1} There is exactly one $\eaxu\in \UBounds{\vara}$ such that
		\[
			\psa\models\phiSup(\disj,\vara,\eaxu)~.
		\]
		%
		%
		%
		\item\label{lemm:phiInf+phiSup2}  If $\eaxu\in \UBounds{\vara}$ and $\psa\models\phiSup(\disj,\vara,\eaxu)$, then 
		\[
			\sem{\psa}{\eaxu} \eeq	\sup\bigl\{\rata\in\QQ\mid\psa[\vara\mapsto\rata]\models\disj\bigr\}~.
		\]
		%
		\item\label{lemm:phiInf+phiSup3} There is exactly one $\eaxl\in\LBounds{\vara}$ such that
		\[
			\psa\models\phiInf(\disj,\vara,\eaxl)~.
		\]
		\item\label{lemm:phiInf+phiSup4}  If $\eaxl\in \LBounds{\vara}$ and $\psa\models\phiInf(\disj,\vara,\eaxl)$, then 
		\[
		\sem{\psa}{\eaxl} \eeq	\inf\bigl\{\rata\in\QQ\mid\psa[\vara\mapsto\rata]\models\disj\bigr\}~.
		\]
	\end{enumerate}
\end{restatable}
\begin{proof}
	See \Cref{proof:lemm:phiInf+phiSup}.
	%
	%
	%
	%
\end{proof}
\noindent
An immediate consequence of \Cref{lemm:phiInf+phiSup} is that, for every $\psa\in\PS$, 
%
\begin{align*}
	\sem{\psa}{\sum\limits_{i=1}^{m_1} \iverson{\phiSup(\disj,\vara,\eaxu_i)} \cdot \eaxu_i}
	\eeq& \sup \bigl\{\rata\in\QQ\mid\psa[\vara\mapsto\rata]\models\disj\bigr\} \\
	 \sem{\psa}{\sum\limits_{i=1}^{m_2} \iverson{\phiInf(\disj,\vara,\eaxl_i)} \cdot \eaxl_i} \eeq& \inf \bigl\{\rata\in\QQ\mid\psa[\vara\mapsto\rata]\models\disj\bigr\}
\end{align*}
whenever $\psa \models \phiSAT(\disj,\vara)$.
It is in this sense that  $\phiSup$ and $\phiInf$ characterize least upper- and greatest lower bounds on $\vara$ admitted by $\disj$.
\subsubsection{Eliminating the Quantitative Quantifiers}
Equipped with the preceding prerequisites, we formalize our construction and prove it sound. Given extended arithmetic expressions $\eaxa,\eaxc\in\ExtLinAX$ and a variable $\vara\in\Vars$, we define
	\begin{align*}
	\eaxc(\vara,\eaxa)
	\eeq 
	\begin{cases}
		\infty & \text{if \vara{} occurs positively in \eaxc{} and } \eaxa=\infty ~\text{or}\\[-1.5mm]
		& \text{\phantom{if }\vara{} occurs negatively in \eaxc{} and } \eaxa=-\infty\\
		-\infty & \text{if \vara{} occurs positively in \eaxc{} and } \eaxa=-\infty ~\text{or}\\[-1.5mm]
		& \text{\phantom{if }\vara{} occurs negatively in \eaxc{} and } \eaxa=\infty\\
		\eaxc & \text{if $\vara\not\in\FVars{\eaxc}$} \\
		\eaxc[\vara/\eaxa] & \text{otherwise}~,
	\end{cases}
\end{align*}
where in the last case we have $\eaxa \in \LinAX$ so $\eaxc[\vara/\eaxa]$ is the standard syntactic replacement\footnote{provided in \Cref{app:syntrepl}.} of $\vara$ by $\eaxa$ in $\eaxc$.
Our sought-after quantifier-free equivalents are\footnote{recall that $\UBounds{\vara} = \{\eaxu_1,\ldots,\eaxu_{m_1}\}$ and $\LBounds{\vara} = \{\eaxl_1,\ldots,\eaxl_{m_2}\}$.}
\begin{align}
	\label{eq:qe_stage3_sup}
	&\QE{\SupOO{\vara}{(\disj\searrow\eaxc)}} \eeq \iverson{\neg \phiSAT(\disj,\vara)}\cdot(-\infty)  \\
	&\qquad{}+\iverson{\phiSAT(\disj,\vara)}\cdot \begin{cases}
		\sum\limits_{i=1}^{m_1} \iverson{\phiSup(\disj,\vara,\eaxu_i)} \cdot \eaxc(\vara,\eaxu_i) & \text{if \vara{} occurs positively in \eaxc}\\
		\sum\limits_{i=1}^{m_2} \iverson{\phiInf(\disj,\vara,\eaxl_i)} \cdot \eaxc(\vara,\eaxl_i)& \text{if \vara{} occurs negatively in \eaxc}\\
		 \eaxc &\text{if $\vara \not\in\FVars{\eaxc}$}
	\end{cases}
	\notag
\end{align}
%
and
\begin{align}
	\label{eq:qe_stage3_inf}
	&\QE{\InfOO{\vara}{(\disj\nearrow\eaxc)}} \eeq \iverson{\neg \phiSAT(\disj,\vara)}\cdot \infty  \\
	&\qquad{}+\iverson{\phiSAT(\disj,\vara)}\cdot \begin{cases}
		\sum\limits_{i=1}^{m_2} \iverson{\phiInf(\disj,\vara,\eaxl_i)} \cdot \eaxc(\vara,\eaxl_i) & \text{if \vara{} occurs positively in \eaxc}\\
		\sum\limits_{i=1}^{m_1} \iverson{\phiSup(\disj,\vara,\eaxu_i)} \cdot \eaxc(\vara,\eaxu_i)& \text{if \vara{} occurs negatively in \eaxc}\\
		\eaxc &\text{if $\vara \not\in\FVars{\eaxc}$}~.
	\end{cases}
	\notag
\end{align}
%
This complies with our intuition from \Cref{ex:running_ex_4}. We apply classical Fourier-Motzkin variable elimination to check whether the respective supremum (infimum) evaluates to $-\infty$ (resp.\ $\infty$). If $\phiSAT(\disj,\vara)$ is satisfied, then we inspect $\eaxc$ closer, select the right bound $\eaxu$ (resp.\ $\eaxl$) on $\vara$ in $\disj$ via $\phiSup$ (resp.\ $\phiInf$), and substitute $\vara$ in $\eaxc$ by $\eaxu$ ($\eaxl$) while obeying the arithmetic laws for the extended reals given in \Cref{sec:quantities}.
The resulting quantities are partitioning and:
\begin{restatable}{theorem}{thmThirdLevelSup}
	\label{lemm:ThirdLevelSup}
	Let $\vara\in\Vars$, $\eaxc\in\ExtLinAX$, and let $\disj$ satisfy the hypotheses of \Cref{lemm:phiSAT}. We have:
	\begin{enumerate}
		\item $\SupOO{\vara}{(\disj\searrow\eaxc)} ~{}\equiv{}~ \QE{\SupOO{\vara}{(\disj\searrow\eaxc)}}$
		\item $\InfOO{\vara}{(\disj\nearrow\eaxc)} ~{}\equiv{}~ \QE{\InfOO{\vara}{(\disj\nearrow\eaxc)}} $
	\end{enumerate}
\end{restatable}
\begin{proof}
	Let $\psa\in\PS$. We distinguish the cases $\psa \models \phiSAT(\disj,\vara)$ and $\psa\not\models \phiSAT(\disj,\vara)$.
	For $\psa \models \phiSAT(\disj,\vara)$, we distinguish $\vara \not\in \FVars{\eaxc}$ and $\vara \in \FVars{\eaxc}$, the latter case being the most interesting. The key insight is that if $\eaxc$ is of the form $\rata_0 + \sum_{\varb\in\Vars}\rata_{\varb} \cdot \varb$ and $\psa\models \phiSAT(\disj,\vara)$, then 
	\begin{align*}
		& \sem{\psa}{\SupOO{\vara}{(\disj\searrow\eaxc)}} \\
		\eeq & \rata_0  + \big( \sum_{\varb\in\Vars\setminus\{\vara\}}\rata_{\varb} \cdot \psa(\varb)  \big)
		+ \rata_\vara \cdot 
		\begin{cases}
			\displaystyle \sem{\psa}{\sum\limits_{i=1}^{m_1} \iverson{\phiSup(\disj,\vara,\eaxu_i)} \cdot \eaxu_i } & \text{if}~\rata_\vara > 0 \\
			\displaystyle \sem{\psa}{\sum\limits_{i=1}^{m_2} \iverson{\phiInf(\disj,\vara,\eaxl_i)} \cdot \eaxl_i }   & \text{if}~\rata_\vara < 0 \
		\end{cases} 
	\end{align*}
	by \Cref{lemm:phiInf+phiSup}. See \Cref{proof:lemm:ThirdLevelSup} for a detailed proof.
\end{proof}

\subsection{Algorithmically Eliminating Quantitative Quantifiers}
\label{sec:algo}
We write $\ELIM{\lqa}$ for the result of the recursive procedure in \Cref{alg:qe}; its input must be partitioning.
\begin{algorithm}[t]
	\SetAlgoLined
	\SetNoFillComment
	\textbf{Input: } partitioning $\lqa\in\LQ$\\
	\textbf{Output: }{quantifier-free} partitioning $\ELIM{\lqa}\in\LQ$ with $\ELIM{\lqa}\equiv\lqa$\\
	%
	%
	\uIf{$\lqa$ is quantifier-free}{
		\Return $\lqa$
	}

	\uElseIf{$\lqa$ is of the form $\quanta \vara \colon \lqc$}{
		$\lqc \gets \ELIM{\lqc} $\; 
		transform $\lqc$ into an equivalent member of $\XGNF{\vara}$, normalizing infinity comparisons first\; \label{algoline:transform_gnf}
		{\color{gray} \tcp{let $\lqc = \sum_{i=1}^n \big[ \bigvee_{j=1}^{m_i} \disj_{i,j} \big] \cdot \eaxa_i$} }
		\uIf{$\quanta = \Sup$}{
			\Return $\MAX \big( \bigcup_{i=1}^n \bigcup_{j=1}^{m_i} 
				\big\{ \underbrace{\QE{\Sup \vara \colon \disj_{i,j} \searrow \eaxa_i}}_{\text{given by \Cref{eq:qe_stage3_sup} on page \pageref{eq:qe_stage3_sup}}}
				\big\}
			\big)$
		}\uElseIf{$\quanta = \Inf$}{
		\Return $\MIN \big( \bigcup_{i=1}^n \bigcup_{j=1}^{m_i} 
		\big\{ 
		\underbrace{\QE{\Inf \vara \colon \disj_{i,j} \nearrow \eaxa_i}}_{\text{given by \Cref{eq:qe_stage3_inf} on page \pageref{eq:qe_stage3_inf}}}
		 \big\}
		\big)$ 
		}
	}
	
	\caption{$\ELIM{\cdot}$ --- Quantitative Quantifier Elimination}
	\label{alg:qe}
\end{algorithm}
We summarize our quantifier elimination technique in \Cref{alg:qe}, which takes as input a partitioning $\lqa \in \LQ$ and computes a quantifier-free equivalent $\ELIM{\lqa}$ of $\lqa$ by proceeding in a recursive inner- to outermost fashion. Since $\lqa$ is partitioning and since both $\MAX$ and $\MIN$ always return partitioning quantities, the transformation of $\lqc$ into $\XGNF{\vara}$ only involves normalizing infinity comparisons, transforming every Boolean expression into DNF, and collecting coefficients to cancel or isolate $\vara$ in each finite inequality. The soundness of \Cref{alg:qe} is an immediate consequence of our observations from the preceding sections. Moreover, the algorithm terminates because the number of recursive invocations equals the number of quantifiers occurring in $\lqa$.

In order to upper-bound the syntactic size of the output of \Cref{alg:qe}, we agree on the following: The size $\bxsize{\bxa}$ of a Boolean expression $\bxa$ is the number of (not necessarily distinct) inequalities it contains. 
The \emph{width} $\lqwidth{\lqa}$ of $\lqa\in\LQ$ is the number of summands in its matrix, and its \emph{depth} $\lqdepth{\lqa}$ is the maximum size of a guard in that matrix. 
\begin{restatable}{theorem}{algoSound}
	\label{thm:algo_sound}
	\Cref{alg:qe} is sound and terminates. Moreover, for partitioning $\lqa\in\LQ$ with $\lqwidth{\lqa} = n$ and $\lqdepth{\lqa}= m$ containing exactly one quantifier,
	%
	\[
		\lqwidth{\ELIM{\lqa}} \leq n\cdot2^m\cdot(m+2)^{n\cdot2^m}
		\quad\text{and}\quad
		\lqdepth{\ELIM{\lqa}} \leq n \cdot 2 ^m \cdot \big((\nicefrac{m+2}{2})^2 + m + 1\big)~.
	\]
\end{restatable}
\begin{proof}
	We exploit that (i) transforming a Boolean expression $\bxa$ of size $l$ into DNF produces at most $2^l$ disjuncts, each consisting of at most $l$ linear inequalities and (ii) if $\disj$ is of size $l$, then $\bxsize{\phiSAT(\disj,\vara)} \leq \left( \nicefrac{l+2}{2} \right)^2$. See \Cref{proof:algo_sound} for details.
\end{proof}
%

For non-partitioning $\lqa$, we obtain the following corollary:
\begin{cor}[Arbitrary Inputs]\label{cor:arbitrary_qe}
For every $\lqa\in\LQ$, the quantity $\EAll{\lqa}$ is quantifier-free, partitioning, and equivalent to $\lqa$. If $\lqa$ has one quantifier, width $n$, and depth $m$, let $N=2^n$ and $M=n\cdot m$. Then
\begin{align*}
 \lqwidth{\EAll{\lqa}}&\leq N\cdot2^M\cdot(M+2)^{N\cdot2^M},\\
 \lqdepth{\EAll{\lqa}}&\leq N\cdot2^M\cdot\bigl((\nicefrac{M+2}{2})^2+M+1\bigr)~.
\end{align*}
\end{cor}
\begin{proof}
Combine the partitioning construction of \Cref{def:partition} with \Cref{thm:algo_sound}. Infinity normalization does not increase the width or the number of inequality occurrences.
\end{proof}

Iterating the bounds of \Cref{thm:algo_sound} yields an upper bound on output width and depth that is non-elementary in the number of quantifiers. This estimate reflects the repeated conversion to DNF and the Cartesian products in the $\MAX$ and $\MIN$ constructions of \Cref{lemm:statewise_maxmin}. These bounds do not establish that such growth is inherent to quantitative quantifier elimination. By comparison, classical quantifier elimination in linear rational arithmetic admits a doubly exponential upper bound on running time in the input-formula length \cite{Weispfenning88}, and hence also on output length.

Classical formula-size \emph{lower bounds} transfer to quantifier-free \emph{partitioning} quantities as follows. A prenex LRA formula embeds into a partitioning quantity with linear overhead by replacing its quantifier-free matrix $\bxa$ with $\iverson{\bxa}\cdot1+\iverson{\neg\bxa}\cdot0$, and existential and universal quantifiers with supremum and infimum quantifiers, respectively. The resulting quantity is $\{0,1\}$-valued and represents the truth value of the original formula. Any equivalent quantifier-free partitioning quantity $\qflqa=\sum_{i=1}^r\iverson{\bxb_i}\cdot\eaxa_i$ is then likewise $\{0,1\}$-valued, so comparing its value with $\nicefrac{1}{2}$ recovers a classical formula equivalent to the original one:
\[
 \bigvee_{i=1}^r\left(\bxb_i\wedge\eaxa_i>\nicefrac{1}{2}\right)~,
\]
where comparisons involving $\infty$ or $-\infty$ simplify to true or false. This introduces at most one additional inequality per summand, giving at most $\sum_{i=1}^r\bxsize{\bxb_i}+r\leq r\cdot(m+1)$ inequality occurrences, where $r=\lqwidth{\qflqa}$ and $m=\lqdepth{\qflqa}$. Thus, a lower bound of $\ell$ on the number of inequality occurrences in any classical quantifier-free equivalent yields $r\cdot(m+1)\geq\ell$. Classical lower bounds constrain the combined measure $r\cdot(m+1)$ (the measures bounded in \Cref{thm:algo_sound}) (rather than width or depth separately). Under matching encodings of affine expressions and rational coefficients, the extracted formula also has total encoding length linear in that of $\qflqa$.

Consequently, lower bounds that hold for \emph{every} equivalent quantifier-free LRA formula, irrespective of its Boolean structure, transfer to quantifier-free partitioning quantities under the corresponding size measure. Bounds restricted to DNF outputs or to a particular elimination algorithm do not suffice for this argument. In particular, since quantifier-free equivalents of certain linear formulas necessarily have doubly exponential length \cite{Weispfenning88}, quantifier-free partitioning equivalents of the corresponding quantities are of doubly exponential size as well. The same lower bounds do not automatically extend to non-partitioning quantities, since making them partitioning can itself cause exponential growth. A tighter analysis of our construction and matching upper and lower bounds for quantitative quantifier elimination are left for future work.

		\section{Quantitative Craig Interpolation}
		\label{sec:craig}
In this section,  we derive a quantitative Craig interpolation theorem from \Cref{alg:qe} and \Cref{thm:algo_sound}. %
Let us first agree on a notion of quantitative Craig interpolants, which is a quantitative analogue of the notion from \cite{10.1145/1181775.1181789}.
Given $\lqa,\lqb \in \LQ$, we say that $\lqa$ (quantitatively) \emph{entails} $\lqb$, denoted by $\lqa \qentails \lqb$, if 
%
	$\forall \psa \in \PS \colon  \sem{\psa}{\lqa} \leq \sem{\psa}{\lqb}$.
%
%
%
\begin{definition}[Quantitative Craig Interpolant]
	Given $\lqa,\lqb,\lqc\in\LQ$ with $\lqa\qentails\lqb$, we say that $\lqc$ is a \emph{(quantitative) Craig interpolant} of $(\lqa,\lqb)$, if 
	\begin{align*}
		\lqa \qentails \lqc ~\text{and}~ \lqc \qentails \lqb
		\qquad\text{and}\qquad
		\FVars{\lqc} \subseteq \FVars{\lqa} \cap \FVars{\lqb}~.
		\tag*{\finishdefinition}
	\end{align*}
\end{definition}
In words, $\lqc$ sits between $\lqa$ and $\lqb$ and the free variables occurring in $\lqc$ also occur free in \emph{both} $\lqa$ and $\lqb$. 
$\lqc$ is \emph{stronger} than $\lqd$ if $\lqc\qentails\lqd$, and \emph{weaker} if $\lqd\qentails\lqc$. Strongest and weakest interpolants are therefore unique up to semantic equivalence.

We will now see that piecewise linear quantities enjoy the property of being \emph{quantifier-free interpolating} \cite{10.1145/1181775.1181789}: For all $\lqa,\lqb \in \LQ$ with $\lqa\qentails\lqb$, there exists a quantifier-free Craig interpolant of $(\lqa,\lqb)$. More precisely, we prove that both the \emph{strongest} and the \emph{weakest} Craig interpolants 
of
 $(\lqa,\lqb)$ are \emph{quantifier-free and effectively constructible}.
Our construction is inspired by existing techniques for constructing \emph{classical} Craig interpolants via \emph{classical} quantifier elimination \cite{DBLP:journals/jsat/EsparzaKS08}: 
By \enquote{projecting-out} the free variables in $\lqa$ which are \emph{not} free in $\lqb$ via $\Sup$, we obtain the \emph{strongest} Craig interpolant of $(\lqa,\lqb)$. Dually, by \enquote{projecting-out} the free variables in $\lqb$ which are \emph{not} free in $\lqa$ via $\Inf$, we obtain the \mbox{\emph{weakest} Craig interpolant of $(\lqa,\lqb)$. Put formally:}
%

\begin{restatable}{theorem}{thmCraig}
	\label{thm:craig}
	Let $\lqa,\lqb \in \LQ$ with $\lqa\qentails\lqb$. We have:
	\begin{enumerate}
		\item\label{thm:craig1} For $\{\vara_1,\ldots,\vara_n\} = \FVars{\lqa} \setminus \FVars{\lqb}$,
		%
		\[
		\lqc \eeq \EAll{\Sup \vara_1 \ldots \Sup \vara_n \colon \lqa}
		\]
		is the strongest quantitative Craig interpolant of $(\lqa,\lqb)$, i.e.,
		\[
		\forall~\text{Craig interpolants $\lqd$ of $(\lqa,\lqb)$} \colon \quad 
		\lqc \qentails \lqd~.
		\]
		\item\label{thm:craig2} For $\{\varb_1,\ldots,\varb_m\} = \FVars{\lqb} \setminus \FVars{\lqa}$,
		%
		\[
		\lqc \eeq \EAll{\Inf \varb_1 \ldots \Inf \varb_m \colon \lqb}
		\]
		is the weakest quantitative Craig interpolant of $(\lqa,\lqb)$, i.e.,
		\[
		\forall~\text{Craig interpolants $\lqd$ of $(\lqa,\lqb)$} \colon \quad 
		\lqd \qentails \lqc~.
		\]
	\end{enumerate}
\end{restatable}
\begin{proof}
	See \Cref{proof:thm:craig}.
\end{proof}
%

%
%
%

%
\begin{example}
	\label{ex:craig2}
	Consider the following quantities $\lqa,\lqa'$ which satisfy $\lqa \qentails \lqa'$:
	\begin{align*}
		 \lqa \eeq& \iverson{\vara \geq 0}\cdot \vara + \iverson{\vara \geq 0 \wedge \varb \leq \vara}\cdot \varb \\
		 \lqa' \eeq & \iverson{\vara \geq 0 \wedge \varc \geq \vara} \cdot (2\cdot\vara + \varc +1) + \iverson{\varc < \vara}\cdot \infty
	\end{align*}
	Pre-processing $\lqa$ and $\lqa'$ to make them partitioning and simplifying yields
	%
	\begin{align}
		 \underbrace{\EAll{\Sup \varb \colon \lqa } \eeq \iverson{\vara \geq 0}\cdot 2\cdot \vara}_{\text{strongest Craig interpolant}}
		 \quad\text{and}\quad
		 \underbrace{\EAll{\Inf \varc \colon \lqb}\eeq \iverson{\vara\geq 0}\cdot (3\cdot \vara + 1)}_{\text{weakest Craig interpolant}}~.
	\tag*{\finishexample}
	\end{align}
\end{example}

		\section{Reasoning about Probabilistic Programs with Unbounded Non-Determinism}
		\label{sec:nondet_programs}
		In this section, we present our results on effectively constructing closed-form representations for expected outcomes of loop-free linear probabilistic programs featuring unbounded non-determinism. Towards this end, we introduce an appropriate programming language in \Cref{sec:pgcl}. We present McIver's \& Morgan's weakest pre-expectation calculus \cite{mciver_morgan} extended by constructs for \emph{unbounded} non-determinism --- a program logic tailored to reasoning about minimal/maximal expected outcomes. In \Cref{sec:wp:compute}, we combine that program logic with our quantifier elimination results, yielding the desired algorithm. Finally, in \Cref{sec:wp:loops}, we illustrate how our approach enables checking inductivity of candidate upper bounds on minimal/maximal expected outcomes of loops.

\subsection{Probabilistic Programs with Unbounded Non-Determinism}
\label{sec:pgcl}
We introduce a loop-free probabilistic programming language à la McIver \& Morgan \cite{mciver_morgan} featuring \emph{unbounded} (pure) non-determinism. Importantly, all arithmetic- and Boolean expressions occurring in a program are assumed to be linear according to \Cref{def:LinAX,def:Bool}:
\begin{definition}\label{def:pgcl}
Linear programs $\pp$ in the (loop-free fragment of the) {\emph{{probabilistic guarded command language}}} $\pgcl$ adhere to the grammar
\begin{align*}
	\pp \quad \longrightarrow \quad& \phantom{{}\mid{}~} \SKIP  \tag{effectless program}\\
	&{}\mid{}~ \ASSIGN{\vara}{\axa}  \tag{assignment}\\
	&{}\mid{}~ \UNDCHOICE{\vara}  \tag{non-deterministic assignment}\\
	&{}\mid{}~ \COMPOSE{\pp}{\pp}  \tag{sequential composition}\\
	&{}\mid{}~ \PCHOICE{\pp}{\proba}{\pp} \tag{probabilistic choice}\\
	&{}\mid{}~ \NDCHOICE{\pp}{\pp} \tag{non-deterministic choice}\\
	&{}\mid{}~ \ITE{\bxa}{\pp}{\pp} \tag{conditional choice}~, 
	%
\end{align*}
where:
\begin{enumerate}
	\item $\axa \in \LinAX{}$ is a {{linear arithmetic expression}}, 
	\item $\proba \in [0,1] \cap \QQ$ is a {rational probability}, and
	\item $\bxa\in \Bool$ is a Boolean expression also referred to as a \emph{guard}. \hfill $\triangle$
\end{enumerate}
\end{definition}
Let us go over each construct. $\SKIP$ does nothing. $\ASSIGN{\vara}{\axa}$ assigns the value of the linear arithmetic expression $\axa$ under the current variable valuation to the variable $\vara$. $\UNDCHOICE{\vara}$ non-deterministically assigns an arbitrary rational number to $\vara$. The sequential composition $\COMPOSE{\pp_1}{\pp_2}$ first executes $\pp_1$ and then $\pp_2$.  $\PCHOICE{\pp_1}{\proba}{\pp_2}$ executes $\pp_1$ with probability $\proba$, and $\pp_2$ with probability $1-\proba$. The binary choice $\NDCHOICE{\pp_1}{\pp_2}$ non-deterministically selects either branch. $\ITE{\bxa}{\pp_1}{\pp_2}$ executes either $\pp_1$ or $\pp_2$, depending on whether the Boolean expression is satisfied by the current variable valuation or not.

\subsection{The Weakest Pre-Expectation Calculus}
\label{sec:wp}
\begin{table}[t]
	\begin{center}
		\begin{tabularx}{0.9\textwidth}{X@{\quad}l@{\quad~~}lX}
			\toprule
			\toprule
			$\boldsymbol{\pp}$ & $\mathsf{\mathbf{awp}}\boldsymbol{\llbracket\pp \rrbracket(\FF)}$  \\[0.5ex]
			\hline 
			$\SKIP$ & $\FF$  \rule{0pt}{3.5ex}\\[1.8ex]
			$\ASSIGN{\vara}{\axa}$ & $\mylambda{\psa} \FF\big( \psa[\vara\mapsto\sem{\psa}{\axa}] \big)$   \\[1.5ex]
			$\UNDCHOICE{\vara}$ & $\mylambda{\psa} \sup \big\{ \FF\big( \psa[\vara\mapsto\rata] \big) ~\mid~ \rata\in\QQ \big\}$   \\[1.5ex]
			$\COMPOSE{\pp_1}{\pp_2}$ & $\awp{\pp_1}{\awp{\pp_2}{\FF}}$ \\[1.5ex]
			$\PCHOICE{\pp_1}{\proba}{\pp_2}$ & $\mylambda{\psa} \proba \cdot \awp{\pp_1}{\FF}(\psa) + (1-\proba )\cdot \awp{\pp_2}{\FF}(\psa)$ \\[1.5ex]
			$\NDCHOICE{\pp_1}{\pp_2}$ & $\mylambda{\psa}\max\{\awp{\pp_1}{\FF}(\psa),\awp{\pp_2}{\FF}(\psa)\}$ \\[1.5ex]
			$\ITE{\bxa}{\pp_1}{\pp_2}$ & $\mylambda{\psa} 
							\begin{cases}
										\awp{\pp_1}{\FF}(\psa) & \text{if}~ \psa \models \bxa \\
										\awp{\pp_2}{\FF}(\psa) & \text{if}~ \psa\models \neg\bxa
								\end{cases}$  \\[3.0ex]
			\bottomrule
			\bottomrule		
\end{tabularx}
	\end{center}%
	\caption{Recursive definition of $\awp{\pp}{\FF}$. The definition of $\dwp{\pp}{\FF}$ is obtained from (i) replacing $\sup$ by $\inf$, (ii) replacing $\max$ by $\min$, and (iii) replacing $\awpsymbol$ by $\dwpsymbol$.}
	\label{tab:wp}%
\end{table}%
The central objects the weakest pre-expectation calculus operates on are so called \emph{expectations} --- functions assigning a non-negative real number or infinity to every variable valuation. The \emph{set of expectations} is 
\[
\E \eeq \{ \FF ~\mid~ \FF \colon \PS \to \PosRealsInf \ \}
\]
and we denote expectations by $\FF$, $\FG$, and variations thereof. Expectations can be thought of as random variables on a program's state space. With this in mind, consider the following:
\begin{definition}
	\label{def:wp}
	Let $\pp \in \pgcl$ and let $\FF\in \E$. Then: 
	\begin{enumerate}
		\item The \emph{angelic weakest pre-expectation \[\awp{\pp}{\FF} \in \E\] of $\pp$ w.r.t.\ post-expectation $\FF$} is defined recursively on the structure of $\pp$ in \Cref{tab:wp}.
		\item The \emph{demonic weakest pre-expectation \[\dwp{\pp}{\FF} \in \E\] of $\pp$ w.r.t.\ post-expectation $\FF$} is defined analogously by (i) replacing $\sup$ by $\inf$, (ii) replacing $\max$ by $\min$, and (iii) replacing $\awpsymbol$ by $\dwpsymbol$ in the rules given in \Cref{tab:wp}. \hfill $\triangle$
	\end{enumerate}
\end{definition}
%
Let us first gain some intuition on the values $\awpsymbol/\dwpsymbol$ determine. Given $\somewpsymbol\in\{\awpsymbol,\dwpsymbol\}$, $\pp \in \pgcl$, and an expectation $\FF\in\E$, $\somewp{\pp}{\FF}$ is yet another expectation assigning a non-negative real or infinity to each variable valuation. Now let $\psa\in\PS$ and think of it as an \emph{initial} variable valuation for $\pp$. Due to the non-deterministic assignments and choices possibly occurring in $\pp$, executing $\pp$ on $\psa$ does not necessarily yield a \emph{unique} distribution over final variable valuations but a (possibly infinite) set of such distributions --- one for each possible way of resolving the non-determinism in $\pp$. Now, the intuition is that
%
%
\[
\awpsymbol/\dwp{\pp}{\FF}(\psa) \eeq
\substack{\text{\normalsize\emph{supremal/infimal} expected value of $\FF$ w.r.t.\ \emph{all} distributions} \\ \text{\normalsize of \emph{final} states reached after executing $\pp$ on \emph{initial} valuation $\psa$~.}}
\]

Both $\dwp{\pp}{\FF}$ and $\awp{\pp}{\FF}$ can be defined recursively on the structure of the program $\pp$, obviating the need for explicitly determining the set of final distributions of~$\pp$. This recursive definition is shown in \Cref{tab:wp}. Let us go over the rules. For that, let $\somewpsymbol\in\{\dwpsymbol,\awpsymbol\}$. Since $\SKIP$ does nothing, $\somewp{\SKIP}{\FF}$ is just $\FF$. $\somewp{\ASSIGN{\vara}{\axa}}{\FF}$ is obtained from \enquote{replacing} $\vara$ in $\FF$ by the right-hand side of the assignment $\axa$. Since, for now, we treat expectations as purely semantic objects, this replacement is formalized semantically. $\somewp{\UNDCHOICE{\vara}}{\FF}$ is the expectation that evaluates to the infimum/supremum of all possible values of $\FF$ obtained from replacing $\vara$ by arbitrary rationals under the current valuation. The rule for sequential composition suggests that $\somewp{\pp}{\FF}$ is obtained in a backward-moving fashion: We first determine $\somewp{\pp_2}{\FF}$ and plug this intermediate pre-expectation into $\somewp{\pp_1}{\cdot}$. $\somewp{\PCHOICE{\pp_1}{\proba}{\pp_2}}{\FF}$ is the weighted sum of the pre-expectations of the two branches, weighted according to the probability of the respective branch being executed. Finally, $\somewp{\ITE{\bxa}{\pp_1}{\pp_2}}{\FF}$ evaluates to the pre-expectation of $\pp_1$ if the current variable valuation satisfies $\bxa$, \mbox{otherwise to the pre-expectation of $\pp_2$.}
\begin{example}
	\label{ex:wp}
	Reconsider program $\pp$ from \Cref{ex:intro:wp}:
	\begin{align*}
		& \UNDCHOICE{\vara_1}\,; \\
		&\PCHOICE{\ASSIGN{\varb_1}{\varb_1 + \vara_1}}{\nicefrac{1}{2}}{\SKIP}\,; \\
		& \UNDCHOICE{\vara_2}\,; \\
		&\PCHOICE{\ASSIGN{\varb_2}{\varb_2 + \vara_2}}{\nicefrac{1}{2}}{\SKIP} 
	\end{align*}
	Now fix an initial variable valuation $\psa$. 
	Due to the purely non-deterministic assignments, executing $\pp$ on $\psa$ does not determine a unique probability distribution over final variable valuations on termination of $\pp$ but an \emph{infinite} set obtained from \emph{all} possible choices for $\vara$. 
	Now recall that we asked:
	\begin{center}
		\emph{Under all possible choices the players can make,} \\
		\emph{what is the minimal (maximal) probability of them ending up at the same position?}
	\end{center}
	To model this question using $\dwpsymbol$/$\awpsymbol$, let $\FF$ be the post-expectation given by
	\[
		\FF = \mylambda{\psa}
		\begin{cases}
			1 & \text{if $\psa(\varb_1)=\psa(\varb_2)$} \\
			0 & \text{otherwise}~,
		\end{cases}
	\]
	which is the indicator function of the event that $\varb_1=\varb_2$ holds on termination of $\pp$. The sought-after minimal probabilities (depending on the initial valuation $\psa$) are then given by
	\begin{align*}
		&\dwp{\pp}{\FF}(\psa) \eeq\\
		&\inf_{\rata_1\in\QQ} 
		\nicefrac{1}{2} \cdot \left(\inf_{\rata_2\in\QQ}
			\nicefrac{1}{2}\cdot \FF(\psa[\varb_1 \mapsto \psa(\varb_1)+\rata_1, \varb_2 \mapsto \psa(\varb_2)+\rata_2]) 
			 + \nicefrac{1}{2}\cdot \FF(\psa[\varb_1 \mapsto \psa(\varb_1)+\rata_1]) \right) \\
		&\qquad{}+{} 
		\nicefrac{1}{2} \cdot \left(\inf_{\rata_2\in\QQ}
			\nicefrac{1}{2}\cdot \FF(\psa[\varb_2 \mapsto \psa(\varb_2)+\rata_2]) 
			+ \nicefrac{1}{2} \cdot \FF(\psa)
		 \right)~,
	\end{align*}
	and similarly for $\awpsymbol$. The algorithm presented in the next section --- so to speak --- eliminates the $\inf$'s, yielding a closed-form representation for $\dwp{\pp}{\FF}$ (and $\awp{\pp}{\FF}$).
	\hfill $\triangle$
\end{example}
\subsection{Effectively Representing Optimal Weakest Pre-Expectations}
\label{sec:wp:compute}
\begin{table}[t]
	\begin{center}
		\begin{tabularx}{0.8\textwidth}{X@{\hspace{0.3em}}l@{\hspace{0.3em}}lX}
			\toprule
			\toprule
			$\boldsymbol{\pp}$ & $\mathsf{\mathbf{\overline{awp}}}\boldsymbol{\llbracket\pp \rrbracket(\lqa)}$  \\[0.5ex]
			\hline 
			$\SKIP$ & $\EAll{\lqa}$  \rule{0pt}{3.5ex}\\[1.8ex]
			$\ASSIGN{\vara}{\axa}$ & $\EAll{\lqa}[\vara/\axa]$   \\[1.5ex]
			$\UNDCHOICE{\vara}$ & $\EAll{\SupOO{\vara}{\lqa}}$   \\[1.5ex]
			$\COMPOSE{\pp_1}{\pp_2}$ & $\xawp{\pp_1}{\xawp{\pp_2}{\lqa}}$ \\[1.5ex]
			$\PCHOICE{\pp_1}{\proba}{\pp_2}$ & $\proba \cdot \xawp{\pp_1}{\lqa} + (1-\proba)\cdot\xawp{\pp_2}{\lqa}$ \\[1.5ex]
			$\NDCHOICE{\pp_1}{\pp_2}$ & $\MAX\{\PART{\xawp{\pp_1}{\lqa}},\PART{\xawp{\pp_2}{\lqa}}\}$ \\[1.5ex]
				$\ITE{\bxa}{\pp_1}{\pp_2}$ & $\iverson{\bxa}\cdot \xawp{\pp_1}{\lqa}  + \iverson{\neg\bxa} \cdot \xawp{\pp_2}{\lqa}$  \\[1.0ex]
			\bottomrule
			\bottomrule		
\end{tabularx}
	\end{center}%
	\caption{Syntactic variants of $\dwpsymbol$/$\awpsymbol$ operating on piecewise linear expectations.}
	\label{tab:xwp}%
\end{table}%
We have just introduced the recursive definitions of weakest pre-expectations for purely semantic functions from variable valuations to numbers. In this section, we present a \emph{syntactic} variant of weakest pre-expectations operating on piecewise linear quantities. Using our quantifier elimination algorithm, this will enable us to obtain effectively constructible, closed-form representations of weakest pre-expectations whenever the post-expectation can be represented as a piecewise linear quantity --- even in the presence of unbounded non-determinism.

Piecewise linear quantities in $\LQ$ generally evaluate to negative numbers. As a first step, we define a decidable subset of piecewise linear quantities that actually represent expectations, i.e., functions from variable valuations to \emph{non-negative} numbers or $\infty$.
\begin{definition}\label{def:LE}
	The decidable set of \emph{piecewise linear expectations} is defined as
	\begin{align*}
		\LE \quad\eeq\quad \{ \lqa \in \LQ \quad {}|{}\quad \mylambda{\psa} \sem{\psa}{\lqa} ~{}\in{}~ \E \}~.
		\tag*{$\triangle$}
	\end{align*}
\end{definition}
Notice that the condition $\mylambda{\psa}\sem{\psa}{\lqa} \in \E $ is equivalent to $\forall \psa\in\PS\colon \sem{\psa}{\lqa} \geq 0$, which is easily seen to be decidable: Given $\lqa\in\LQ$, let 
\[
	\EAll{\lqa} \eeq \SNF
\]
be the quantifier-free equivalent of $\lqa$ obtained from \Cref{alg:qe}. Since $\EAll{\lqa}$ is partitioning (cf.\ \Cref{def:GNF}.\ref{def:GNF1}), $\forall \psa\in\PS\colon \sem{\psa}{\EAll{\lqa}} \geq 0$ is equivalent to 
\[
	\bxa_i \wedge \eaxa_i < 0~\text{being unsatisfiable for all $i \in \{1,\ldots,n\}$}~.
\]
The latter can be decided via SMT solving over linear rational arithmetic (LRA).
\begin{remark}
 For loop-free programs, a signed extension would have to ensure that every intermediate backward calculation is well-defined. In particular, at each probabilistic choice, branch values with positive weights must not contribute opposite infinities at the same valuation. We do not develop this extension here and retain the standard nonnegative setting. However, in the presence of loops, as considered in \Cref{sec:wp:loops}, signed expectations require additional technical machinery to handle integrability issues~\cite{KaminskiK17}. We therefore retain the standard setting of non-negative expectations.\hfill$\triangle$
\end{remark}

Piecewise linear expectations are syntactic representations of  (semantic) expectations as introduced in the previous section. This yields syntactic variants of both $\dwpsymbol$ and $\awpsymbol$:
\begin{definition}
	\label{def:xwp}
	Let $\pp \in \pgcl$ and let $\lqa\in \LE$. Then: 
	\begin{enumerate}
		\item The \emph{syntactic} angelic weakest pre-expectation \[\xawp{\pp}{\lqa} ~{}\in{}~ \LE\] of $\pp$ w.r.t.\  $\lqa$ is defined recursively on the structure of $\pp$ by the rules in \Cref{tab:xwp}.
		\item The \emph{syntactic} demonic weakest pre-expectation \[\xdwp{\pp}{\lqa} ~{}\in{}~ \LE\] of $\pp$ w.r.t.\ $\lqa$ is defined analogously by (i) replacing $\Sup$ by $\Inf$, (ii) $\MAX$ by $\MIN$, and (iii) $\xawpsymbol$ by $\xdwpsymbol$. \hfill $\triangle$
	\end{enumerate}
\end{definition}
The rules in \Cref{tab:xwp} closely follow their semantic counterparts in \Cref{tab:wp}. To keep the definition concise, we introduce some straightforward shorthands: $\EAll{\lqa}[\vara/\axa]$ is the (quantifier-free) piecewise linear expectation obtained from $\EAll{\lqa}$ by replacing every occurrence of $\vara$ by $\axa$. Moreover, given a quantifier-free
\[
	\lqc \eeq \SNF ~{}\in{}~ \LE~,
\]
and a rational probability $\proba$, we let 
\[
	\proba \cdot \lqc \eeq \sum\limits_{\idxa=1}^{\nna}\iverson{\bxa_\idxa}\cdot (\proba\cdot\eaxa_\idxa)~,
\]
where $\proba\cdot\eaxa_\idxa$ is the extended arithmetic expression $\eaxa_\idxa$ scaled by $\proba$. Finally, we recall that 
\[
	\iverson{\bxa}\cdot\lqc \eeq \sum_{i=1}^n \iverson{\bxa \wedge \bxa_i}\cdot \eaxa_i~.
\]
By performing quantifier elimination (\Cref{alg:qe}) whenever we encounter a piecewise linear expectation possibly containing quantifiers, we obtain only quantifier-free expressions which moreover represent the sought-after pre-expectations. Let us put this more formally:
\begin{theorem}
	\label{thm:compute_wp}
	Let $\somewpsymbol \in \{\awpsymbol,\dwpsymbol\}$. For all $\pp\in\pgcl$ and all $\lqa \in \LE$, we have
	\begin{enumerate}
		\item $\xsomewp{\pp}{\lqa}$ belongs to $\LE$ and is quantifier-free, and
		\item $\xsomewp{\pp}{\lqa}$ is sound, i.e., $\mylambda{\psa} \sem{\psa}{\xsomewp{\pp}{\lqa}} = \somewp{\pp}{\mylambda{\psa} \sem{\psa}{\lqa}}$.
	\end{enumerate}
	In particular, $\somewp{\pp}{\mylambda{\psa} \sem{\psa}{\lqa}}$ is a computable, $\PosRatsInf$-valued function.
\end{theorem}
\begin{proof}
	By induction on $\pp$, proving soundness, quantifier-freeness, and membership in $\LE$ using \Cref{cor:arbitrary_qe,lem:syntrepl}. For non-deterministic assignment and every $\psa\in\PS$, we have
\begin{align*}
 \sem{\psa}{\xawp{\UNDCHOICE{\vara}}{\lqa}}
 &\eeq \sem{\psa}{\EAll{\SupOO{\vara}{\lqa}}}\\
 &\eeq \sup_{\rata\in\QQ}\sem{\psa[\vara\mapsto\rata]}{\lqa}\\
 &\eeq \awp{\UNDCHOICE{\vara}}{\mylambda{\psa}\sem{\psa}{\lqa}}(\psa)~.
\end{align*}
The infimum case is analogous. Both operations preserve nonnegativity, and their outputs are quantifier-free by elimination. The remaining cases use finite affine substitution, composition, rational weighted sums, complementary guards, and, for binary non-deterministic choice, \Cref{lemm:statewise_maxmin} applied to partitioning operands. These operations preserve nonnegativity and quantifier-freeness. 
\end{proof}
Recursively applying the rules in \Cref{tab:xwp} thus yields the sought-after algorithm.
\begin{example}
	\label{ex:wp:computed}Reconsider the program $\pp$ and the post-expectation $\FF$ from \Cref{ex:wp}.
	%
	%
	%
	With $\xdwpsymbol/\xawpsymbol$, we can compute quantifier-free piecewise linear expectations mapping \emph{all initial positions $\varb_1,\varb_2$} to the sought-after minimal/maximal probabilities. For that, let \[\lqa \eeq \iverson{\varb_1 = \varb_2}\] be the (syntactic) post-expectation corresponding to $\FF$. Write $\pp=\pp_1;\pp_2$, where, for $i\in\{1,2\}$,
	\[
	 \pp_i\eeq\UNDCHOICE{\vara_i}\,;\,
	 \PCHOICE{\ASSIGN{\varb_i}{\varb_i+\vara_i}}{\nicefrac{1}{2}}{\SKIP}~.
	\]
	We apply the rules of \Cref{tab:xwp}, beginning with $\pp_2$. The probabilistic choice yields the following post-expectation for its preceding non-deterministic assignment:
	\[
	 \frac{1}{2}\cdot\iverson{\varb_1=\varb_2+\vara_2}
	 +\frac{1}{2}\cdot\iverson{\varb_1=\varb_2}~.
	\]
	For every fixed pair of positions, choosing $\vara_2=\varb_1-\varb_2$ makes the first indicator $1$, whereas any different rational choice makes it $0$. Eliminating the corresponding infimum or supremum quantifier therefore gives
	\[
	 \xdwp{\pp_2}{\lqa}\equiv\frac{1}{2}\cdot\iverson{\varb_1=\varb_2}
	 \qquad\text{and}\qquad
	 \xawp{\pp_2}{\lqa}\equiv\frac{1}{2}+\frac{1}{2}\cdot\iverson{\varb_1=\varb_2}~.
	\]
	Using these intermediate expectations for $\pp_1$ yields
	\begin{align*}
	 \xdwp{\pp}{\lqa}
	 &\equiv\EAll{\InfOO{\vara_1}{(
	   \frac{1}{4}\cdot\iverson{\varb_1+\vara_1=\varb_2}
	   +\frac{1}{4}\cdot\iverson{\varb_1=\varb_2}
	)}}~,\\
	 \xawp{\pp}{\lqa}
	 &\equiv\EAll{\SupOO{\vara_1}{(
	   \frac{1}{2}+\frac{1}{4}\cdot\iverson{\varb_1+\vara_1=\varb_2}
	   +\frac{1}{4}\cdot\iverson{\varb_1=\varb_2}
	 )}}~.
	\end{align*}
	The infimum is attained by any $\vara_1\neq\varb_2-\varb_1$, and the supremum by $\vara_1=\varb_2-\varb_1$. Thus, after simplification, we get
	\[
		\xdwp{\pp}{\lqa} \eeq \iverson{\varb_1 = \varb_2} \cdot \frac{1}{4}
		\qquad\text{and}\qquad
		\xawp{\pp}{\lqa} \eeq  \frac{3}{4} + \iverson{\varb_1 = \varb_2} \cdot \frac{1}{4}~.
	\]
	In words, the sought-after \emph{minimal} probability is $\nicefrac{1}{4}$ whenever the players start at the same position, and $0$ otherwise. On the other hand, the sought-after \emph{maximal} probability is $1$ whenever the players start at the same position, and $\nicefrac{3}{4}$, otherwise. 
	%
	\hfill $\triangle$
\end{example}
\subsection{Checking Inductivity of Supplied Upper Bounds for Loops}
\label{sec:wp:loops}
We now illustrate how our approach enables automatically checking inductivity of candidate upper bounds on weakest pre-expectations of \emph{loops}. For simplicity, we restrict to loops with loop-free bodies. An extension for multiple nested or sequenced loops is straightforward yet conceptually orthogonal to the presented approach.

Now fix a loop $\WHILEDO{\bxa}{\pp}$ with loop-free $\pp$ throughout this section. Given $\FF\in\E$ and $\somewpsymbol\in\{\awpsymbol,\dwpsymbol\}$, we define the \emph{loop characteristic function} (w.r.t.\ $\FF$) as
\[
 \Phi \colon \E \to \E, \quad 
 \Phi_{\FF}(\FG)
 \eeq \mylambda{\psa}
 \begin{cases}
    \somewp{\pp}{\FG}(\psa)&\text{if }\psa\models\bxa,\\
    \FF(\psa)&\text{otherwise}~,
 \end{cases}
\]
where $\somewpsymbol\in\{\awpsymbol,\dwpsymbol\}$. As is standard \cite{mciver_morgan}, we then define
\[
\somewp{\WHILEDO{\bxa}{\pp}}{\FF}\eeq\mu\Phi_{\FF}~,
\]
where $\mu\Phi_{\FF}$ denotes the least fixed point of $\Phi_{\FF}$. This least fixed point is guaranteed to exist since $\Phi_{\FF}$ is a monotone function on the complete lattice $(\E, \, \leq)$ of expectations, where $\leq$ is the pointwise lifted order on $\PosRealsInf$. For $\mu\Phi_{\FF}$, \emph{Park induction} applies:
\begin{theorem}[\cite{kaminski_diss}]
   \label{thm:park}
 Let $\FG \in \E$. If $\Phi_\FF(\FG) \leq \FG$, then $\mu\Phi_{\FF} \leq \FG$.
\end{theorem}

If $\Phi_\FF(\FG) \leq \FG$, then we say that $\FG$ is \emph{inductive}. Hence, Park induction enables establishing upper bounds on weakest pre-expectations of loops by reasoning about \emph{a single} loop iteration. Establishing lower bounds requires additional side conditions \cite{DBLP:journals/pacmpl/HarkKGK20} and is not considered here. 

Our techniques imply that inductivity is decidable for piecewise linear expectations:
\begin{theorem}[Decidability of Checking Inductivity]\label{cor:loop_bound}
Let $\pp$ be loop-free, $\bxa\in\Bool$, $\lqa,\lqc\in\LE$, and write $\FF = \mylambda{\psa}\sem{\psa}{\lqa}$ and $\FG = \mylambda{\psa}\sem{\psa}{\lqc}$. It is decidable whether
\[
  \Phi_\FF(\FG) ~{}\leq{}~ \FG
\]
holds, i.e., whether $\FG$ is inductive.
\end{theorem}
\begin{proof}
First compute $\lqa_0\eeq\EAll{\lqa}$ and $\lqc_0\eeq\EAll{\lqc}$. By \Cref{cor:arbitrary_qe}, $\lqa_0$ and $\lqc_0$ are quantifier-free expectations equivalent to $\lqa$ and $\lqc$, respectively. Hence, 
\[
\Phi_\FF(\FG) ~{}\leq{}~ \FG
\qquad\text{iff}\qquad
 \iverson{\neg\bxa}\cdot\lqa_0
 +\iverson{\bxa}\cdot\xsomewp{\pp}{\lqc_0}
 \quad\qentails\quad\lqc_0~.
\] 
Now write $\sum_i\iverson{\bxb_i}\cdot \eaxa_i = \PART{\iverson{\neg\bxa}\cdot\lqa_0
 +\iverson{\bxa}\cdot\xsomewp{\pp}{\lqc_0}}$ and $\sum_j\iverson{\bxc_j}\cdot\eaxb_j = \lqc_0$. Since both quantities are partitioning, entailment holds exactly when 
 \begin{align}
 \Norm{\bigvee_{i,j}(\bxb_i\wedge\bxc_j\wedge\eaxa_i>\eaxb_j)}
 \tag{see \cite[Section 4]{batz_diss}}
\end{align}
is unsatisfiable, which is a decidable LRA problem. The equivalences $\lqa_0\equiv\lqa$ and $\lqc_0\equiv\lqc$, together with soundness of $\xsomewpsymbol$ (\Cref{thm:compute_wp}) now imply the claim, i.e., the above effectively constructible LRA formula is decidably unsatisfiable iff $\Phi_\FF(\FG) \leq \FG$.
\end{proof}

\begin{example}[An Upper Bound for a Probabilistic Loop]
\label{ex:loop_bound}
Consider the loop
\begin{align*}
 &\WHILE{\vara>0} \\
   &\qquad \UNDCHOICE{\varb}\,;\\
   &\qquad\PCHOICE{\ASSIGN{\vara}{\varb}}{\nicefrac{1}{2}}{\ASSIGN{\vara}{-1}} \\
 &\}
\end{align*}
and let $\pp$ denote its body. Each iteration chooses an arbitrary rational value for $\varb$ and then, with equal probability, either assigns that value to $\vara$ or sets $\vara$ to $-1$. We seek an upper bound on the maximal probability of terminating with $\vara=0$. For this, take the quantifier-free post-expectation and candidate upper bound
\[
 \lqa \eeq \iverson{\vara=0}
 \qquad\text{and}\qquad
 \lqc \eeq \iverson{\vara=0}+\iverson{\vara>0}\cdot\frac{1}{2}~,
\]
which both belong to $\LE$. Write $\FF=\mylambda{\psa}\sem{\psa}{\lqa}$ and $\FG=\mylambda{\psa}\sem{\psa}{\lqc}$. We apply the reduction from the proof of \Cref{cor:loop_bound} with $\bxa=(\vara>0)$ and $\somewpsymbol=\awpsymbol$:
\begin{enumerate}
\item \emph{Partition the input expectations.}
Compute $\lqa_0=\EAll{\lqa}$ and $\lqc_0=\EAll{\lqc}$. Since the inputs are quantifier-free, these calls only partition them. Simplifying guards and discarding unsatisfiable cells gives
\begin{align*}
 \lqa_0 &{}\equiv{} \iverson{\vara=0}\cdot1+\iverson{\vara\neq0}\cdot0~,\\
 \lqc_0 &{}\equiv{} \iverson{\vara<0}\cdot0+\iverson{\vara=0}\cdot1
                   +\iverson{\vara>0}\cdot\frac{1}{2}~.
\end{align*}

\item \emph{Compute the body transformer.}
By the assignment and probabilistic-choice rules in \Cref{tab:xwp}, the pre-expectation of the probabilistic choice is
\begin{align*}
 \qflqa ~{}={}~& \frac{1}{2}\cdot\lqc_0[\vara/\varb]
       +\frac{1}{2}\cdot\lqc_0[\vara/-1]\\
 ~{}\equiv{}~& \iverson{\varb<0}\cdot0+\iverson{\varb=0}\cdot\frac{1}{2}
          +\iverson{\varb>0}\cdot\frac{1}{4}~.
\end{align*}
The preceding non-deterministic assignment therefore yields
\begin{align*}
 \xawp{\pp}{\lqc_0}~{}\equiv{}~& \EAll{\SupOO{\varb}{\qflqa}}\\
 ~{}\equiv{}~& \iverson{\true}\cdot\frac{1}{2}~,
\end{align*}
since the three cells are nonempty and the supremum of their values $0$, $\nicefrac{1}{2}$, and $\nicefrac{1}{4}$ is $\nicefrac{1}{2}$, attained at $\varb=0$.

\item \emph{Partition the left-hand side of the entailment.}
Substituting the computed body transformer gives
\begin{align*}
 &\PART{\iverson{\vara\leq0}\cdot\lqa_0
        +\iverson{\vara>0}\cdot\xawp{\pp}{\lqc_0}}\\
 ~{}\equiv{}~& \iverson{\vara<0}\cdot0+\iverson{\vara=0}\cdot1
                +\iverson{\vara>0}\cdot\frac{1}{2}~.
\end{align*}
Index these three summands from left to right by $(\bxb_i,\eaxa_i)$, and those of $\lqc_0$ by $(\bxc_j,\eaxb_j)$, as in the proof of \Cref{cor:loop_bound}.

\item \emph{Check the LRA formula for a counterexample.}
The reduction asks whether
\[
 \bxa_{\mathrm{bad}}\eeq
 \Norm{\bigvee_{i,j=1}^{3}
  (\bxb_i\wedge\bxc_j\wedge\eaxa_i>\eaxb_j)}
\]
is unsatisfiable. All operands are finite, so infinity normalization leaves this formula unchanged. Of the nine disjuncts, the six with $i\neq j$ have contradictory guards. The remaining three give
\begin{align*}
 \bxa_{\mathrm{bad}}
 ~{}\equiv{}~& (\vara<0\wedge0>0)
       \vee(\vara=0\wedge1>1)
       \vee\left(\vara>0\wedge\frac{1}{2}>\frac{1}{2}\right)\\
 ~{}\equiv{}~& \false~.
\end{align*}
Thus, the LRA formula is unsatisfiable, establishing $\Phi_\FF(\FG)\leq\FG$.
\end{enumerate}
By Park induction (\Cref{thm:park}), the maximal probability of terminating with $\vara=0$ is therefore at most $\nicefrac{1}{2}$ for every initial valuation with $\vara>0$.
\hfill $\finishexample$
\end{example}

		\section{Conclusion}
		\label{sec:concl}
		We have investigated both quantitative quantifier elimination and quantitative Craig interpolation for piecewise linear quantities --- an assertion language in automatic quantitative software verification. We have provided a quantifier elimination algorithm that terminates on every well-formed input and returns an equivalent quantifier-free quantity, and analyzed syntactic output-size bounds. Using our algorithm, we have derived a quantitative Craig interpolation theorem for arbitrary \mbox{piecewise linear quantities.} 

We see ample space for future work. First, we could investigate alternative quantifier elimination procedures: Our algorithm can be understood as a quantitative generalization of Fourier-Motzkin variable elimination \cite{motzkin1936beitraege,fourier1825analyse}. It would be interesting to apply, e.g., virtual substitution \cite{DBLP:journals/cj/LoosW93} in the quantitative setting and to compare the so-obtained approaches --- both empirically and theoretically. We might also benefit from improvements of Fourier-Motzkin variable elimination such as \toolfont{FMplex} \cite{DBLP:journals/corr/abs-2310-00995} to improve the practical feasibility of our approach.  Moreover, we have focussed on $\QQ$-valued variables. We plan to investigate techniques which apply to integer-valued variables using, e.g., Cooper's method \cite{cooper} and in how far our results can be generalized to a non-linear setting. We also plan to apply our results from \Cref{sec:nondet_programs} to the technique from \cite{DBLP:journals/pacmpl/BatzKRW25} for approximate inference of continuous probabilistic programs with hard guarantees.

Finally, we plan to investigate potential applications of our techniques:
\begin{enumerate}
	\item Dillig et al.\ \cite{DBLP:conf/oopsla/DilligDLM13} present a quantifier elimination-based algorithm for generating inductive loop invariants of classical programs abductively. Generalizing this algorithm to the probabilistic setting, where weakest pre\emph{conditions} are replaced by weakest pre\emph{expectations}, might yield a promising application of our quantifier elimination algorithm.
	\item We are currently investigating the applicability of McMillan’s interpolation and Boolean satisfiability (SAT)-based model checking algorithm \cite{DBLP:conf/cav/McMillan03} to \emph{probabilistic} program verification. One of the major challenges is to obtain suitable quantitative interpolants and we hope that our results on the existence of interpolants spark the development of suitable techniques.
	\item In the light of the above application and \Cref{rem:simple_interpolants}, we plan to adapt Albarghouthi's and McMillan's technique for computing \cite{DBLP:conf/cav/AlbarghouthiM13}  \enquote{simpler} interpolants.
\end{enumerate}

	\bibliographystyle{alphaurl}
	\bibliography{literature}
	
	\allowdisplaybreaks
	\appendix
	\section{Auxiliary Results}

\subsection{Construction of Valuation-Wise Minima}
\label{app:minima}
Use the notation of \Cref{lemm:statewise_maxmin}: $\seta=\{\qflqa_1,\ldots,\qflqa_\nna\}$, where $\qflqa_i=\sum_{j=1}^{\nnb_i}\iverson{\bxa_{i,j}}\eaxa_{i,j}$ is quantifier-free and partitioning, and $\underline{\nnb}_i=\{1,\ldots,\nnb_i\}$. We define
	\begin{align*}
	&\MIN(\seta)  = \sum\limits_{(\idxc_1,\ldots,\idxc_\nna)\in \underline{\nnb}_1\times\ldots\times\underline{\nnb}_\nna}\quad\sum\limits_{\idxa=1}^{\nna}\\
	&\qquad \qquad\qquad\Bigl[
	\underbrace{\bigwedge\limits_{\nnc=1}^{\nna} \bxa_{\nnc,\idxc_\nnc}}_{
		\text{$\qflqa_\nnc$ evaluates to $\eaxa_{\nnc,\idxc_\nnc}$}
	}
	\land
	\underbrace{\bigwedge\limits_{\nnc= 1}^{\idxa - 1} \eaxa_{\idxa,\idxc_\idxa}<\eaxa_{\nnc,\idxc_\nnc}
		\land
		\bigwedge\limits_{\nnc= \idxa + 1}^{\nna} \eaxa_{\idxa,\idxc_\idxa}\leq \eaxa_{\nnc,\idxc_\nnc}}_{\substack{\text{$\qflqa_\idxa$ is the quantity with smallest index} \\ \text{evaluating to the sought-after minimum}}}
	\Bigr] \cdot \eaxa_{\idxa,\idxc_\idxa}~.
\end{align*}

\subsection{Syntactic Replacement of Variables by Expressions}
\label{app:syntrepl}
Let $\axc=\ratb_0+\sum_{i=1}^{\cardnum}\ratb_i\cdot\vara_i\in\LinAX$ be a \emph{finite} affine expression and let $\eaxa\in\ExtLinAX$. For $\eaxa=\rata_0+\sum_{i=1}^{\cardnum}\rata_i\cdot\vara_i$, define
\[
 \syntRepl{\vara_j}{\eaxa}{\axc}
 \eeq (\rata_0+\rata_j\cdot\ratb_0)
       +\sum_{\substack{i=1\\i\neq j}}^{\cardnum}
          (\rata_i+\rata_j\cdot\ratb_i)\cdot\vara_i
       +\rata_j\cdot\ratb_j\cdot\vara_j~.
\]
For $\eaxa\in\{-\infty,\infty\}$, let $\syntRepl{\vara_j}{\eaxa}{\axc}=\eaxa$. Thus the original coefficient of the replaced variable is removed, but occurrences contributed by the replacement expression remain. For example, $\syntRepl{\vara}{\vara}{0}=0$ and $\syntRepl{\vara}{\vara}{(\vara+1)}=\vara+1$.

\begin{lemma}[Substitution]\label{lem:syntrepl}
For $\axc\in\LinAX$, $\eaxa\in\ExtLinAX$, $\psa\in\PS$, and $j\in\{1,\ldots,\cardnum\}$,
\[
 \sem{\psa}{\syntRepl{\vara_j}{\eaxa}{\axc}}
 =\sem{\psa[\vara_j\mapsto\sem{\psa}{\axc}]}{\eaxa}~.
\]
The analogous substitution property holds for Boolean expressions and quantifier-free quantities. For quantifier-free quantities, substitution preserves membership in $\LQ$, and preserves partitioning and nonnegativity whenever present.
\end{lemma}
\begin{proof}
For finite $\eaxa$, expand the right-hand side and collect the coefficient of every variable; the displayed definition is exactly the resulting affine expression. Infinite $\eaxa$ is unaffected. Induction on Boolean syntax gives the corresponding satisfaction equivalence. Applying these identities summand-wise proves the quantity identity. Any simultaneous occurrence of opposite active infinities after substitution would, under the updated rational valuation, already violate \Cref{def:LQ}.\ref{def:LQ3} before substitution. The same valuation identity preserves partitioning and nonnegativity.
\end{proof}
This definition concerns ordinary substitution by a finite affine expression. The notation $\eaxc(\vara,\eaxa)$ of \Cref{sec:stage3} additionally handles infinite endpoints by the sign of the coefficient; it is not substitution of an infinite value into a rational valuation.

	\section{Omitted Proofs}

\subsection{Proof of \Cref{Nlemm:FirstLevelSup}}
\label{Nproof:lemm:FirstLevelSup}
\NlemFirstLevSup*
\begin{proof} We have
	\begin{align*}
		&\sem{\psa}{\SupOO{\vara}{\SNF}} \\
		\eeq&\sup\Bigl\{\sem{\psa[\vara\mapsto\rata]}{\SNF} ~{}\mid{}~\rata\in\QQ\Bigr\}\tag{\Cref{def:SemOfLQ}}\\
		\eeq&\sup\Bigl\{\sem{\psa[\vara\mapsto\rata]}{[\bxa_\idxa]\cdot\eaxa_\idxa}~{}\mid{}~\rata\in\QQ,\idxa\in\{1,\ldots,\nna\}\text{ s.t. }\psa[\vara\mapsto\rata]\models\bxa_\idxa\Bigr\}\tag{quantity is partitioning}\\
		\eeq&\sup\Bigl(\bigcup\limits_{\idxa\in\{1,\ldots,\nna\}}\Bigl\{\sem{\psa[\vara\mapsto\rata]}{[\bxa_\idxa]\cdot\eaxa_\idxa}~{}\mid{}~\rata\in\QQ\text{ s.t. }\psa[\vara\mapsto\rata]\models\bxa_\idxa\Bigr\}\Bigr)
		\tag{rewrite set}\\
		\eeq&\max\Bigl(\bigcup\limits_{\idxa\in\{1,\ldots,\nna\}}\Bigl\{\sup\Bigl\{\sem{\psa[\vara\mapsto\rata]}{[\bxa_\idxa]\cdot\eaxa_\idxa}~{}\mid{}~\rata\in\QQ\text{ s.t. }\psa[\vara\mapsto\rata]\models\bxa_\idxa\Bigr\}\Bigr\}\Bigr)
		\tag{supremum of finite union is maximum of individual suprema}
		\\
		\eeq&\max\Bigl(\bigcup\limits_{\idxa\in\{1,\ldots,\nna\}}\Bigl\{\sup\Bigl\{\sem{\psa[\vara\mapsto\rata]}{\down{\bxa_\idxa}{\eaxa_\idxa}}\\
		&\qquad\qquad\qquad\qquad\qquad{}~{}\mid{}~\rata\in\QQ\text{ s.t. }\psa[\vara\mapsto\rata]\models\bxa_\idxa\Bigr\}\Bigr\}\Bigr)\tag{adding zero}\\
		\eeq&\max\Bigl(\bigcup\limits_{\idxa\in\{1,\ldots,\nna\}}\Bigl\{\sup\Bigl\{\sem{\psa[\vara\mapsto\rata]}{\down{\bxa_\idxa}{\eaxa_\idxa}}\mid\rata\in\QQ\Bigr\}\Bigr\}\Bigr)\tag{$-\infty$ is neutral w.r.t.\ $\sup$}\\
		\eeq&\max\Bigl(\bigcup\limits_{\idxa\in\{1,\ldots,\nna\}}\Bigl\{\sup\Bigl\{\sem{\psa[\vara\mapsto\rata]}{\bxa_\idxa\searrow\eaxa_\idxa}\mid\rata\in\QQ\Bigr\}\Bigr\}\Bigr)
		\tag{by definition}
		\\
		\eeq&\max\Bigl(\bigcup\limits_{\idxa\in\{1,\ldots,\nna\}}\bigl\{\sem{\psa}{\SupOO{\vara}{(\bxa_\idxa\searrow\eaxa_\idxa)}}\bigr\}\Bigr)\tag{\Cref{def:SemOfLQ}}\\
		\eeq&\max\bigl\{\sem{\psa}{\SupOO{\vara}{(\bxa_\idxa\searrow\eaxa_\idxa)}}\mid\idxa\in\{1,\ldots,\nna\}\bigr\}
		\tag{rewrite set}
	\end{align*}
	The reasoning for $\Inf$ is analogous.
\end{proof}

\subsection{Proof of \Cref{lemm:ExploitDNF}}
\label{proof:lemm:ExploitDNF}
\lemmExploitDNF*
\begin{proof}
	For the $\Sup$-case, consider the following:
		\begin{align*}
			&\sem{\psa}{\SupOO{\vara}{(\bxa\searrow\eaxa)}} \\
			\eeq&\sup\Bigl\{\sem{\psa[\vara\mapsto\rata]}{\bxa\searrow\eaxa}\mid\rata\in\QQ\Bigr\}\tag{\Cref{def:SemOfLQ}}\\
			\eeq&\sup\Bigl\{\sem{\psa[\vara\mapsto\rata]}{\down{\bxa}{\eaxa}}\mid\rata\in\QQ\Bigr\}
			\tag{by definition}
			\\
			\eeq&\sup\Bigl\{\sem{\psa[\vara\mapsto\rata]}{\down{\bxa}{\eaxa}}~{}\mid{}~\rata\in\QQ\text{ s.t. }\psa[\vara\mapsto\rata]\models\bxa\Bigr\}
			\tag{$-\infty$ is neutral w.r.t.\ $\sup$}\\
			\eeq&\sup\Bigl\{\sem{\psa[\vara\mapsto\rata]}{[\bxa]\cdot\eaxa}~{}\mid{}~\rata\in\QQ\text{ s.t. }\psa[\vara\mapsto\rata]\models\bxa\Bigr\}\tag{dropping zero}\\
			\eeq&\sup\Bigl\{\sem{\psa[\vara\mapsto\rata]}{[\bxa]\cdot\eaxa}~{}\mid{}~\rata\in\QQ\text{ s.t. }\psa[\vara\mapsto\rata]\models\disj_1~\text{or}~\ldots~\text{or}~\psa[\vara\mapsto\rata]\models\disj_\nna\Bigr\}
			\tag{disjunctive shape of \bxa}\\
			\eeq&\sup\Bigl(\bigcup\limits_{\idxa\in\{1,\ldots,\nna\}}\Bigl\{\sem{\psa[\vara\mapsto\rata]}{[\bxa]\cdot\eaxa}~{}\mid{}~\rata\in\QQ\text{ s.t. }\psa[\vara\mapsto\rata]\models\disj_\idxa\Bigr\}\Bigr)
			\tag{rewrite set}\\
			\eeq&\max\Bigl(\bigcup\limits_{\idxa\in\{1,\ldots,\nna\}}\Bigl\{\sup\Bigl\{\sem{\psa[\vara\mapsto\rata]}{[\bxa]\cdot\eaxa}~{}\mid{}~\rata\in\QQ\text{ s.t. }\psa[\vara\mapsto\rata]\models\disj_\idxa\Bigr\}\Bigr\}\Bigr)
			\tag{supremum of finite union is maximum of individual suprema}
			\\
			\eeq&\max\Bigl(\bigcup\limits_{\idxa\in\{1,\ldots,\nna\}}\Bigl\{\sup\Bigl\{\sem{\psa[\vara\mapsto\rata]}{[\disj_\idxa]\cdot\eaxa} \\
			&\qquad\qquad\qquad \qquad\qquad~{}\mid{}~\rata\in\QQ\text{ s.t. }\psa[\vara\mapsto\rata]\models\disj_\idxa\Bigr\}\Bigr\}\Bigr) 
			\tag{disjunctive shape of $\bxa$} \\
			\eeq&\max\Bigl(\bigcup\limits_{\idxa\in\{1,\ldots,\nna\}}\Bigl\{\sup\Bigl\{\sem{\psa[\vara\mapsto\rata]}{\down{\disj_\idxa}{\eaxa}} \\
			&\qquad\qquad\qquad \qquad\qquad~{}\mid{}~\rata\in\QQ\text{ s.t. }\psa[\vara\mapsto\rata]\models\disj_\idxa\Bigr\}\Bigr\}\Bigr)\tag{adding zero}\\
			\eeq&\max\Bigl(\bigcup\limits_{\idxa\in\{1,\ldots,\nna\}}\Bigl\{\sup\Bigl\{\sem{\psa[\vara\mapsto\rata]}{\down{\disj_\idxa}{\eaxa}}~{}\mid{}~\rata\in\QQ\Bigr\}\Bigr\}\Bigr)\tag{$-\infty$ is neutral w.r.t.\ $\sup$}\\
			\eeq&\max\Bigl(\bigcup\limits_{\idxa\in\{1,\ldots,\nna\}}\Bigl\{\sup\Bigl\{\sem{\psa[\vara\mapsto\rata]}{\disj_\idxa\searrow\eaxa}~{}\mid{}~\rata\in\QQ\Bigr\}\Bigr\}\Bigr)
			\tag{by definition}
			\\
			\eeq&\max\Bigl(\bigcup\limits_{\idxa\in\{1,\ldots,\nna\}}\bigl\{\sem{\psa}{\SupOO{\vara}{(\disj_\idxa\searrow\eaxa)}}\bigr\}\Bigr)\tag{\Cref{def:SemOfLQ}}\\
			\eeq&\max\bigr\{\sem{\psa}{\SupOO{\vara}{(\disj_\idxa\searrow\eaxa)}}\mid\idxa\in\{1,\ldots,\nna\}\bigr\}
			\tag{rewrite set}
	\end{align*}
	The reasoning for $\Inf$ is analogous.
\end{proof}

\subsection{Proof of \Cref{lemm:phiInf+phiSup}}
\label{proof:lemm:phiInf+phiSup}
\thmPhiInfSup*
%
\begin{proof}
Recall the feasible fibre from \Cref{sec:stage3},
\[
 \Feas{\psa}{\disj,\vara}\eeq
 \{\rata\in\QQ\mid\psa[\vara\mapsto\rata]\models\disj\}~.
\]
This set is non-empty by hypothesis. The finite, non-empty sets $\UBounds{\vara}$ and $\LBounds{\vara}$ have, respectively, an evaluated minimum and maximum. The defining comparisons of $\phiSup$ and $\phiInf$ select exactly the smallest index attaining these values. This proves Statements \ref{lemm:phiInf+phiSup1} and \ref{lemm:phiInf+phiSup3}.

For the endpoint assertions, define
\begin{align*}
 \mina&=\sup\{\sem{\psa}{\eaxc}\mid\eaxc\in\Bounds{>}{\vara}{\disj}\},&
 \minb&=\sup\{\sem{\psa}{\eaxc}\mid\eaxc\in\Bounds{\geq}{\vara}{\disj}\},\\
 \maxa&=\inf\{\sem{\psa}{\eaxc}\mid\eaxc\in\Bounds{<}{\vara}{\disj}\},&
 \maxb&=\inf\{\sem{\psa}{\eaxc}\mid\eaxc\in\Bounds{\leq}{\vara}{\disj}\},
\end{align*}
using $\sup\emptyset=-\infty$ and $\inf\emptyset=\infty$. Since the $\vara$-free literals hold,
\[
 \Feas{\psa}{\disj,\vara}
 =\{\rata\in\QQ\mid\mina<\rata,\ \minb\leq\rata,\ \rata<\maxa,\ \rata\leq\maxb\}~.
\]
In particular, all points of this interval are rational, never infinite. Put $\ell=\max\{\mina,\minb\}$ and $u=\min\{\maxa,\maxb\}$. Non-emptiness gives $\ell\leq u$, $\ell<\infty$, and $u>-\infty$. If $\ell=u$, the interval contains precisely their finite rational common value. If $\ell<u$, every rational strictly between these endpoints satisfies all bounds. Density of $\QQ$, and its unboundedness when an endpoint is infinite, give
\[
 \inf\Feas{\psa}{\disj,\vara}=\ell,
 \qquad
 \sup\Feas{\psa}{\disj,\vara}=u~.
\]
Thus the selected evaluated upper bound equals the supremum of the feasible fibre, and the selected lower bound equals its infimum. This proves Statements \ref{lemm:phiInf+phiSup2} and \ref{lemm:phiInf+phiSup4}.
\end{proof}

\subsection{Proof of \Cref{lemm:ThirdLevelSup}}
\label{proof:lemm:ThirdLevelSup}
\thmThirdLevelSup*
\begin{proof}
	We prove the claim for the $\Sup$-quantifier. The proof for $\Inf$ is dual. Fix some $\psa\in\PS$. 
	Since the supremum of the empty set is $-\infty$, we have
	\begin{align*}
		& \sem{\psa}{\SupOO{\vara}{(\disj\searrow\eaxc)}} \eeq
		\sup_{\rata\in\QQ} \{\sem{\psa\statesubst{\vara}{\rata}}{\eaxc} ~\mid~ \psa\statesubst{\vara}{\rata}\models \disj\} ~.
	\end{align*}
	If $\psa \not\models \phiSAT(\disj,\vara)$, then
	\[
	\sem{\psa}{\SupOO{\vara}{(\disj\searrow\eaxc)}}
	\quad{}={}\quad
	-\infty
	\quad{}={}\quad
	\sem{\psa}{\QE{\SupOO{\vara}{(\disj\searrow\eaxc)}}}
	\]
	by \Cref{lemm:phiSAT}. Now assume $\psa \models \phiSAT(\disj,\vara)$. We distinguish the cases $\vara \not\in \FVars{\eaxc}$ and $\vara \in \FVars{\eaxc}$. If $\vara \not\in \FVars{\eaxc}$, then 
	\begin{align*}
		& \sem{\psa}{\SupOO{\vara}{(\disj\searrow\eaxc)}} \eeq
		\sup_{\rata\in\QQ} \{\sem{\psa}{\eaxc} ~\mid~ \psa\statesubst{\vara}{\rata}\models \disj\} 
		\eeq 
		\sem{\psa}{\phiSAT(\disj,\vara) \searrow \eaxc}~.
	\end{align*}
	Conversely, if $\vara \in \FVars{\eaxc}$, then $\eaxc$ is of the form $\rata_0 + \sum_{\varb\in\Vars}\rata_{\varb} \cdot \varb$ and we have
	\begin{align*}
		& \sem{\psa}{\SupOO{\vara}{(\disj\searrow\eaxc)}} \\
		\eeq  & \sup_{\rata\in\QQ} \{\sem{\psa\statesubst{\vara}{\rata}}{\eaxc} ~\mid~ \psa\statesubst{\vara}{\rata}\models \disj\} \\
		\eeq & \sup_{\rata\in\QQ} \{\rata_0 + \rata_\vara\cdot \rata +  \sum_{\varb\in\Vars\setminus\{\vara\}}\rata_{\varb} \cdot \psa(\varb) ~\mid~ \psa\statesubst{\vara}{\rata}\models \disj\} 
		\tag{semantics of $\eaxc$ under $\psa$}\\
		\eeq & 
		\rata_0  +  \sum_{\varb\in\Vars\setminus\{\vara\}}\rata_{\varb} \cdot \psa(\varb)+
		\sup_{\rata\in\QQ} \{   \rata_\vara\cdot \rata  ~\mid~ \psa\statesubst{\vara}{\rata}\models \disj\} 
		\tag{pull constants out of supremum, set non-empty by assumption}
		\\
		\eeq & \rata_0  + \big( \sum_{\varb\in\Vars\setminus\{\vara\}}\rata_{\varb} \cdot \psa(\varb)  \big)
		+ \rata_\vara \cdot 
		\begin{cases}
			\displaystyle \sup_{\rata\in\QQ} \{ \rata ~\mid~ \psa\statesubst{\vara}{\rata}\models \disj\}  & \text{if}~\rata_\vara > 0 \\
			\displaystyle \inf_{\rata\in\QQ} \{ \rata ~\mid~ \psa\statesubst{\vara}{\rata}\models \disj\}  & \text{if}~\rata_\vara < 0 \
		\end{cases}
		\tag{pull out $\rata_\vara$ by introducing case distinction} \\
		\eeq & \rata_0  + \big( \sum_{\varb\in\Vars\setminus\{\vara\}}\rata_{\varb} \cdot \psa(\varb)  \big)
		+ \rata_\vara \cdot 
		\begin{cases}
			\displaystyle \sem{\psa}{\sum\limits_{i=1}^{m_1} \iverson{\phiSup(\disj,\vara,\eaxu_i)} \cdot \eaxu_i } & \text{if}~\rata_\vara > 0 \\
			\displaystyle \sem{\psa}{\sum\limits_{i=1}^{m_2} \iverson{\phiInf(\disj,\vara,\eaxl_i)} \cdot \eaxl_i }   & \text{if}~\rata_\vara < 0 \
		\end{cases} 
		\tag{\Cref{lemm:phiInf+phiSup}}
		\\
		\eeq &
		\begin{cases}
			\displaystyle \sum\limits_{i=1}^{m_1} \sem{\psa}{\iverson{\phiSup(\disj,\vara,\eaxu_i)}}\cdot \big(
		 \rata_0  + \big( \sum_{\varb\in\Vars\setminus\{\vara\}}\rata_{\varb} \cdot \psa(\varb)  \big)
		+ \rata_\vara   \cdot \sem{\psa}{\eaxu_i}  
			\big)& \text{if}~\rata_\vara > 0 \\
			\displaystyle \sum\limits_{i=1}^{m_2} \sem{\psa}{\iverson{\phiInf(\disj,\vara,\eaxl_i)}}\cdot \big(
			\rata_0  + \big( \sum_{\varb\in\Vars\setminus\{\vara\}}\rata_{\varb} \cdot \psa(\varb)  \big)
			+ \rata_\vara   \cdot \sem{\psa}{\eaxl_i}  
			\big)& \text{if}~\rata_\vara < 0
		\end{cases}
		\tag{ \Cref{lemm:phiInf+phiSup}.\ref{lemm:phiInf+phiSup1} and  \Cref{lemm:phiInf+phiSup}.\ref{lemm:phiInf+phiSup3}} \\
		\eeq & 
			\begin{cases}
			\sem{\psa}{\displaystyle\sum\limits_{i=1}^{m_1} \iverson{\phiSup(\disj,\vara,\eaxu_i)} \cdot \eaxc(\vara,\eaxu_i) }&  \text{if \vara{} occurs positively in \eaxc} \\
			\sem{\psa}{\displaystyle \sum\limits_{i=1}^{m_2} \iverson{\phiInf(\disj,\vara,\eaxl_i)} \cdot \eaxc(\vara,\eaxl_i)}&  \text{if \vara{} occurs negatively in \eaxc}~.
		\end{cases} 
		\tag{by definition of $\eaxc(\vara,\eaxu_i)$ (resp.\ $\eaxc(\vara,\eaxl_i)$)} \\
		\eeq & \sem{\psa}{\QE{\SupOO{\vara}{(\disj\searrow\eaxc)}}}~.
		\tag{by definition}
	\end{align*}
	This completes the proof.
	%
	%
	%
\end{proof}

\subsection{Proof of \Cref{thm:craig}}
\label{proof:thm:craig}
\thmCraig*
\begin{proof}
	We prove \Cref{thm:craig}.\ref{thm:craig1}. The proof of \Cref{thm:craig}.\ref{thm:craig2} is analogous.
		
		We first prove that $\lqc$ is a Craig interpolant of $\lqa$ and $\lqb$.
        Partitioning and elimination introduce no free variables, and elimination removes every quantified variable. Thus
        \[
        \FVars{\lqc}\subseteq\FVars{\lqa}\setminus\{\vara_1,\ldots,\vara_n\}
        =\FVars{\lqa}\cap\FVars{\lqb}~.
        \]
        By \Cref{cor:arbitrary_qe}, $\lqc$ is quantifier-free and equivalent to the displayed quantified quantity. 
		For $\lqa \qentails \lqc$, consider the following for an arbitrary valuation $\psa \in \PS$:
		\begin{align*}
			& \sem{\psa}{\lqa} \lleq \sem{\psa}{\lqc} \\
			\text{if}\quad &  \sem{\psa}{\lqa} \lleq  \sem{\psa}{\Sup \vara_1 \ldots \Sup \vara_n \colon \lqa} \\
			\text{if}\quad & \sem{\psa}{\lqa} \lleq  \sup \{\sem{\psa[\vara_1\mapsto \rata_1,\ldots,\vara_n\mapsto\rata_n]}{\lqa} 
			~\mid~ \rata_1,\ldots,\rata_n \in \QQ\} 
			\tag{\Cref{def:SemOfLQ}}\\
			\text{if}\quad & \text{true by choosing $\rata_1 = \psa(\vara_1),\ldots,\rata_n =\psa(\vara_n)$}~.
		\end{align*}
		For $\lqc \qentails \lqb$, consider the following for an arbitrary valuation $\psa \in \PS$:
		\begin{align*}
			&  \sem{\psa}{\lqc} \lleq \sem{\psa}{\lqb} \\
			\text{if} \quad & \sem{\psa}{\Sup \vara_1 \ldots \Sup \vara_n \colon \lqa} \lleq \sem{\psa}{\lqb}
			\tag{\Cref{thm:algo_sound}} \\
			\text{if}\quad & \sup \{\sem{\psa[\vara_1\mapsto \rata_1,\ldots,\vara_n\mapsto\rata_n]}{\lqa} 
			~\mid~ \rata_1,\ldots,\rata_n \in \QQ\}   \lleq \sem{\psa}{\lqb}
			\tag{\Cref{def:SemOfLQ}} \\
			\text{if}\quad &  \forall{\rata_1,\ldots,\rata_n\in\QQ}\colon 
			\sem{\psa[\vara_1\mapsto \rata_1,\ldots,\vara_n\mapsto\rata_n]}{\lqa} \lleq \sem{\psa}{\lqb}
			\tag{property of suprema} \\
			\text{if}\quad & \forall{\rata_1,\ldots,\rata_n\in\QQ}\colon 
			\sem{\psa[\vara_1\mapsto \rata_1,\ldots,\vara_n\mapsto\rata_n]}{\lqa} \lleq \sem{\psa[\vara_1\mapsto \rata_1,\ldots,\vara_n\mapsto\rata_n]}{\lqb}
			\tag{$\vara_1,\ldots,\vara_n \not\in \FVars{\lqb}$} \\ 
			\text{if}\quad & \lqa \qqentails \lqb~.
			\tag{holds by assumption}
		\end{align*}

		Now let $\lqd$ be an arbitrary Craig interpolant of $\lqa$ and $\lqb$. To prove that $\lqc$ is the strongest Craig interpolant of $\lqa$ and $\lqb$, we prove $\lqc \qentails \lqd$. For that, consider the following for an arbitrary valuation $\psa \in \PS$:
		\begin{align*}
			& \sem{\psa}{\lqc} \lleq \sem{\psa}{\lqd} \\
			\text{if}\quad & \sem{\psa}{\Sup \vara_1 \ldots \Sup \vara_n \colon \lqa} \lleq \sem{\psa}{\lqd} \\
			\text{if}\quad & \sup \{\sem{\psa[\vara_1\mapsto \rata_1,\ldots,\vara_n\mapsto\rata_n]}{\lqa} 
			~\mid~ \rata_1,\ldots,\rata_n \in \QQ\}   \lleq \sem{\psa}{\lqd}
			\tag{\Cref{def:SemOfLQ}} \\
			\text{if}\quad &  \forall{\rata_1,\ldots,\rata_n\in\QQ}\colon 
			\sem{\psa[\vara_1\mapsto \rata_1,\ldots,\vara_n\mapsto\rata_n]}{\lqa} \lleq \sem{\psa}{\lqd}
			\tag{property of suprema} \\
			\text{if}\quad &  \forall{\rata_1,\ldots,\rata_n\in\QQ}\colon 
			\sem{\psa[\vara_1\mapsto \rata_1,\ldots,\vara_n\mapsto\rata_n]}{\lqa} \lleq \sem{\psa[\vara_1\mapsto \rata_1,\ldots,\vara_n\mapsto\rata_n]}{\lqd}
			\tag{$\vara_1,\ldots,\vara_n\not\in\FVars{\lqc'}$}\\ 
			\text{if} \quad &  \lqa \qentails \lqd~.
			\tag{holds by assumption}
		\end{align*}
\end{proof}

\subsection{Proof of \Cref{thm:algo_sound}}
\label{proof:algo_sound}

We rely on the following auxiliary results:
\begin{lemma}
	\label{lem:size_minmax}
	Let $\seta=\{\qflqa_1,\ldots,\qflqa_\nna\} \subseteq \LQ$ for some $\nna\geq 1$, where each $\qflqa_i$
	%
	is partitioning, quantifier-free, and we have 
	\[
		\max \{ \lqwidth{\qflqa_i} ~\mid~ i \in \{1,\ldots,\nna\}\} \leq z
		\quad\text{and}\quad
		\max \{ \lqdepth{\qflqa_i} ~\mid~ i \in \{1,\ldots,\nna\}\} \leq k~.
	\]
	Then:
	\begin{enumerate}
		\item $\lqwidth{\MAX(\seta)} \leq z^\nna \cdot \nna$ and $\lqwidth{\MIN(\seta)} \leq z^\nna \cdot \nna$.
		\item $\lqdepth{\MAX(\seta)} \leq \nna\cdot (k + 1)$ and $\lqdepth{\MIN(\seta)} \leq \nna\cdot (k + 1)$
	\end{enumerate}
\end{lemma}

\begin{lemma}\label{lem:size_third_level}
Let $\vara\in\Vars$, $\eaxa\in\ExtLinAX$, and let $\disj=\bigwedge_{i=1}^n\lit_i$ satisfy the hypotheses of \Cref{lemm:phiSAT}. Let $\lqc$ denote either the right-hand side of \Cref{eq:qe_stage3_sup} or that of \Cref{eq:qe_stage3_inf}. Then
\[
 \lqwidth{\lqc}\leq n+2,
 \qquad
 \lqdepth{\lqc}\leq(\nicefrac{n+2}{2})^2+n~.
\]
\end{lemma}
\begin{proof}
Suppose $\disj$ contains $a$ lower-bound literals, $b$ upper-bound literals, and $r$ $\vara$-free literals, where $a+b+r=n$. Including the default bounds, the feasibility formula has at most
\[
 (a+1)(b+1)+r
 \leq \left(\frac{n-r+2}{2}\right)^2+r
 \leq \left(\frac{n+2}{2}\right)^2
\]
inequality occurrences. The final inequality follows since the difference between the last two bounds is $r(2n-r)/4\geq0$. There are at most $n+1$ candidate bounds. Their selectors each contain at most $n$ comparisons; the infeasible case contributes one further summand. The asserted width and depth bounds follow. Infinity normalization can only decrease these bounds.
\end{proof}

\algoSound*
\begin{proof}
	After transforming $\lqa$ into $\XGNF{\vara}$ (l.\ \ref{algoline:transform_gnf} of \Cref{alg:qe}), we have 
	\[
		\lqwidth{\lqa}  \leq n 
		\quad\text{and}\quad 
		\lqdepth{\lqa} \leq 2^m \cdot m~.
	\]
	Hence, the supremum case combines the quantities from \eqref{eq:qe_stage3_sup} (the infimum case uses \eqref{eq:qe_stage3_inf}) in the form
	\begin{align*}
		\MAX \big( \bigcup_{i=1}^n \bigcup_{j=1}^{m_i} 
		\big\{ 
		\QE{\Sup \vara \colon \disj_{i,j} \searrow \eaxa_i}
		\big\}
		\big)
		\quad\text{where}\quad
		m_i\leq2^m\quad\text{and}\quad\bxsize{\disj_{i,j}}\leq m~.
	\end{align*}
	Hence, by \Cref{lem:size_third_level},
	\begin{align*}
		\lqwidth{\QE{\Sup \vara \colon \disj_{i,j} \searrow \eaxa_i}} \leq m+2
		~\text{and}~
		\lqdepth{\QE{\Sup \vara \colon \disj_{i,j} \searrow \eaxa_i}} \leq (\nicefrac{m+2}{2})^2 + m~.
	\end{align*}
	The claim then follows by \Cref{lem:size_minmax}.
\end{proof}

	\end{document}